\documentclass[11pt]{article}

\newif\ifcomments
\commentstrue

\usepackage{fullpage}
\usepackage[utf8]{inputenc}

\usepackage{graphicx}
\usepackage{amsmath,amsthm,amssymb}
\usepackage{algorithm}
\usepackage{algpseudocode}
\usepackage{multirow}
\usepackage{makecell}
\usepackage{svg} 

\usepackage{thm-restate,color,xspace}
\usepackage{comment}
\usepackage{thmtools}
\usepackage{xcolor}
\usepackage{nameref}
\usepackage{array}
\usepackage{authblk}
\usepackage[bottom]{footmisc}

\definecolor{ForestGreen}{rgb}{0.1333,0.5451,0.1333}
\definecolor{DarkRed}{rgb}{0.65,0,0}
\definecolor{Red}{rgb}{1,0,0}
\usepackage[linktocpage=true,
pagebackref=true,colorlinks,
linkcolor=DarkRed,citecolor=ForestGreen,
bookmarks,bookmarksopen,bookmarksnumbered]
{hyperref}
\usepackage{cleveref}

\usepackage{tikz}
\usetikzlibrary{positioning}

\declaretheorem[numberwithin=section]{theorem}
\declaretheorem[numberlike=theorem]{lemma}

\declaretheorem[numberlike=theorem]{corollary}

\declaretheorem[numberlike=theorem]{conjecture}

\declaretheorem[numberlike=theorem]{claim}

\declaretheorem[numberlike=theorem,style=definition]{definition}
\declaretheorem[numberlike=theorem,style=definition]{remark}

\global\long\def\polylog{\mathrm{polylog}}

\newcommand{\Otil}{\tilde{O}}
\newcommand{\Ohat}{\hat{O}}
\newcommand{\Omegahat}{\hat{\Omega}}

\global\long\def\APVC{\mathrm{APVC}}
\global\long\def\fclique{\mathrm{4clique}}
\global\long\def\dem{\mathrm{dem}}

\global\long\def\Thetahat{\hat{\Theta}}

\newcommand{\ignore}[1]{}

\ifcomments
\def\thatchaphol#1{\marginpar{$\leftarrow$\fbox{TS}}\footnote{$\Rightarrow$~{\sf\textcolor{purple}{#1 --Thatchaphol}}}}
\def\benyu#1{\marginpar{$\leftarrow$\fbox{BW}}\footnote{$\Rightarrow$~{\sf\textcolor{red}{#1 --Benyu}}}}
\def\yaowei#1{\marginpar{$\leftarrow$\fbox{YL}}\footnote{$\Rightarrow$~{\sf\textcolor{blue}{#1 --Yaowei}}}}
\def\zhiyi#1{\marginpar{$\leftarrow$\fbox{ZH}}\footnote{$\Rightarrow$~{\sf\textcolor{olive}{#1 --Zhiyi}}}}

\else
\newcommand{\benyu}[1]{}
\newcommand{\zhiyi}[1]{}
\newcommand{\thatchaphol}[1]{}
\newcommand{\yaowei}[1]{}
\fi 

\begin{document}

\title{Tight Conditional Lower Bounds for Vertex Connectivity Problems}

\author[1]{Zhiyi Huang}
\author[2]{Yaowei Long}
\author[2]{Thatchaphol Saranurak\thanks{Supported by NSF CAREER grant 2238138.}}
\author[1]{Benyu Wang}
\affil[1]{Tsinghua University}
\affil[2]{University of Michigan}

\date{April, 2023}
\maketitle


\pagenumbering{gobble}
\begin{abstract}
We study the fine-grained complexity of graph connectivity problems in unweighted undirected graphs. Recent development shows that all variants of \emph{edge connectivity} problems, including single-source-single-sink, global, Steiner, single-source, and all-pairs connectivity, are solvable in $m^{1+o(1)}$ time, collapsing the complexity of these problems into the almost-linear-time regime. While, historically, \emph{vertex connectivity} has been much harder, the recent results showed that both single-source-single-sink and global vertex connectivity can be solved in $m^{1+o(1)}$ time, raising the hope of putting all variants of vertex connectivity problems into the almost-linear-time regime too. 

We show that this hope is impossible, assuming conjectures on finding 4-cliques. Moreover, we essentially settle the complexity landscape by giving tight bounds for \emph{combinatorial} algorithms in dense graphs. There are three separate regimes:
\begin{enumerate}
\item all-pairs and Steiner vertex connectivity have complexity $\hat{\Theta}(n^{4})$,
\item single-source vertex connectivity has complexity $\hat{\Theta}(n^{3})$, and
\item single-source-single-sink and global vertex connectivity have complexity $\hat{\Theta}(n^{2})$.
\end{enumerate}
For graphs with general density, we obtain tight bounds of $\hat{\Theta}(m^{2})$, $\hat{\Theta}(m^{1.5})$, $\hat{\Theta}(m)$, respectively, assuming Gomory-Hu trees for element connectivity can be computed in almost-linear time.
\end{abstract}
\pagebreak

\pagenumbering{arabic}

\section{Introduction}

Vertex connectivity and edge connectivity are central concepts in graph theory. In an unweighted undirected graph $G$ with $n$ vertices and $m$ edges, the \emph{vertex connectivity }(\emph{edge connectivity}) between two vertices $u,v$ is the maximum number of internally vertex-disjoint (edge-disjoint) paths from $u$ to $v$. Efficient algorithms for variants of both problems have been extensively studied for at least half a century \cite{FF56,Kle69,GH61,GR98,CKMST11,She13,KLOS14,Mad16,van2020bipartite,CKL+22}.

Until very recently, the graph algorithm community has reached a very satisfying conclusion on the complexity of edge connectivity problems: \emph{all variants of edge connectivity problems can be solved in almost-linear time}. There are five variants of the connectivity problems studied in the literature, including (1) global, (2) single-source single-sink, (3) Steiner, (4) single-source, and (5) all-pairs connectivity. See the detailed definitions in \Cref{sect:prelim}. Among these problems, all-pairs connectivity is the hardest via straightforward reductions. For the global edge connectivity problem, Karger showed in 2000 that the problem admits a near-linear time algorithm \cite{Kar00}. In the recent line of work \cite{LP20,AKT21,LP21,AKT22,LPS22,AKL+22}, Abboud et al.~\cite{AKL+22} finally showed that even the all-pairs edge connectivity case could be reduced to the single-source single-sink case, which solvable in almost-linear time via the recent max flow result by Chen et al. \cite{CKL+22}. This finally puts all five problems on edge connectivity into the almost-linear time regime and raises the following hope:
\begin{center}
\emph{Can we solve all variants of vertex connectivity problems in almost-linear time too? }
\par\end{center}

Historically, the vertex connectivity problems have been much more difficult than their edge connectivity counterpart. Furthermore, Abboud et al. \cite{AKT20} showed that the all-pairs vertex connectivity in \emph{weighted} graphs with $\Otil(n)$ edges requires $\Omegahat(n^{3})$ time assuming SETH.\footnote{We use $\Otil(\cdot)$ to hide a $\polylog(n)$ factor and $\Ohat(\cdot)$ to hide a $n^{o(1)}$ factor.} In \emph{directed} unweighted graphs, Abboud et al.~\cite{abboud2018faster} also showed that the problem even requires $\Omegahat(n^{4})$ for combinatorial algorithms and $\Omegahat(n^{\omega +1 })$ time for general algorithms. Thus, at least in weighted or directed graphs, the problem does not admit almost-linear algorithms. 

But, again, the recent development on vertex connectivity in the standard \emph{unweighted undirected} graphs raises the hope for almost-linear time algorithms. Li et al.~\cite{LNP+21} showed how to compute global vertex connectivity using polylogarithmic calls to max flows, which implied an $\Ohat(m)$-time algorithm via the max flow algorithm of \cite{CKL+22} and improved upon  the known $\Otil(mn)$ bound by \cite{HRG00}. Indeed, the recent max flow algorithm \cite{CKL+22} also implies a $\Ohat(m)$-time algorithm for the $(s,t)$-vertex connectivity problem. 
Until now, there are still no nontrivial lower bounds for Steiner, single-source, and all-pairs vertex connectivity in unweighted graphs, and the technique of \cite{AKT20} is quite specific for weighted graphs. The complexity landscape for vertex connectivity problems is still very unclear.

\subsection{Our Results}

We give a firm negative answer to the above open problem. Based on well-known conjectures, we settle the complexity landscape for all five vertex connectivity problems in dense graphs by giving tight bounds for \emph{combinatorial} algorithms, which generally refer to algorithms that do not use fast matrix multiplication. We can obtain tight bounds even in general graphs by assuming a possible hypothesis about computing Gomory-Hu trees for element connectivity. Below, we discuss our results in more detail.

\paragraph{All-pairs vertex connectivity (\Cref{sect:APVClowerbound}).} 
Our first result is a conditional lower bound of the all-pair vertex connectivity (APVC) problem based on the \emph{4-clique conjecture}, which postulates that the running time for deciding the existence of a 4-clique in a graph must be at least $\hat{\Omega}(n^{4})$ time for combinatorial algorithms \cite{BGL17} and $\Omegahat(n^{\omega(1,2,1)})=\Omega(n^{3.25})$ time for general algorithms \cite{DV22}.\footnote{$\omega(1, 2, 1)$ is the exponent of multiplying an $n\times n^{2}$ matrix by an $n^{2}\times n$ matrix.} 

\begin{theorem}
Assuming the 4-clique conjecture, the all-pairs vertex connectivity problem on an undirected unweighted graph with $n$ vertices requires $\Omegahat(n^4)$ time for combinatorial algorithms, and  $\Omega(n^{3.25})$ time for general algorithms.
\label{thm:1.1}
\end{theorem}
\Cref{thm:1.1} gives the first super cubic lower bounds for all-pairs vertex connectivity problems in the standard undirected case. Moreover, the bound is tight for \emph{combinatorial} algorithms. Indeed, a naive algorithm is to call max flow $O(n^{2})$ times, one for each pair of vertices, which takes $\Ohat(mn^{2}) = \Ohat(n^4)$ time. 

It is interesting to compare APVC with the \emph{all-pairs shortest paths} (APSP) problem, a central problem in fine-grained complexity. It is conjectured that the right complexity of APSP in weighted graphs is $\Thetahat(n^3)$ (see e.g.~\cite{WW10}).
\Cref{thm:1.1} shows that, for general algorithms, APVC in unweighted graphs is strictly harder than APSP in weighted graphs, assuming $\omega > 2$, 

\paragraph{Steiner vertex connectivity (\Cref{sect:SteinerLowerBound}).} 
Next, we study the \emph{Steiner vertex connectivity} problem.
In this problem, given a set of vertices $T$, we need to compute the minimum vertex connectivity among all pairs of vertices in $T$.
Even though Steiner vertex connectivity looks much easier than APVC, we can extend the same lower bound
of \Cref{thm:1.1} for APVC to work for the Steiner case too. 

Towards this hardness, we propose a variant of the 4-clique conjecture, the \textit{edge-universal} 4-clique conjecture, which postulates that, given any subset of demand edges  $E_{\dem}\subseteq E(G)$, checking if every edge in $E_{\dem}$ is contained in a 4-clique requires $\Omegahat(n^4)$ time for combinatorial algorithms (see \Cref{conj:AllEdge4Clique} for details).\footnote{Note that this problem is at least as hard as the problem when $E_{\dem} = E(G)$, i.e. when we need to check if every edge is contained in a 4-clique.}

\begin{theorem}
Assuming the edge-universal 4-clique conjecture, the Steiner vertex connectivity problem on an undirected unweighted graph with $n$ vertices requires $\Omegahat(n^4)$ time for combinatorial algorithms.
\label{thm:SteinerVertexConn}
\end{theorem}
We note that our reduction would imply hardness for general algorithms too if we conjectured the hardness of the edge-universal 4-clique problem for general algorithms.

\paragraph{Upper bounds for general density (\Cref{sect:UpperBound}).} 
On sparse graphs, we observe that there is still some discrepancy between our lower bounds and the naive algorithm. 
More precisely, we observe that our lower bounds for combinatorial algorithms for all-pairs and Steiner vertex connectivity can be easily extended to $\Omegahat(m^2)$ in a graph with $m$ edges. 
However, the naive algorithm requires $\Ohat(mn^2)$ time. For example, in sparse graphs, the naive algorithm takes $\Ohat(n^3)$ time, while our lower bound is $\Omegahat(n^2)$. 

We fix the above discrepancy by showing improved algorithms in sparse graphs via the following result.
\begin{theorem}
Given an $n$-vertex $m$-edge unweighted graph, there is an algorithm to solve the APVC problem in $\hat{O}(m^{11/5})$ time with high probability. Assuming that the \emph{element connectivity Gomory-Hu tree} can be constructed in $\Ohat(m)$ time, the running time can be improved to $\hat{O}(m^2)$.
\label{thm:SparseFastAlgo}
\end{theorem}

\Cref{thm:SparseFastAlgo} gives the first subcubic algorithm for computing all-pairs vertex connectivity in sparse graphs. Moreover, the bound $\Thetahat(m^2)$ is tight with our lower bounds for all density regimes, assuming the almost-linear-time construction of the element connectivity Gomory-Hu tree.

We note that the \emph{element connectivity Gomory-Hu tree} is a generalization of the well-known \emph{edge connectivity Gomory-Hu tree} \cite{GH61}, whose almost-linear-time construction was very recently shown by Abboud \textit{et al.} \cite{AKL+22}. It is quite believable that an element connectivity Gomory-Hu tree can also be constructed in almost-linear time. 

\paragraph{The complexity landscape of vertex connectivity.}

Based on these main results, we obtain some other results and corollaries on all variants of vertex connectivity problems, as summarized in \Cref{table}. In contrast to the edge connectivity problems which can all be solved in almost-linear time, there are three separate regimes for vertex connectivity. 
\begin{enumerate}
\item all-pairs and Steiner vertex connectivity have complexity $\Thetahat(n^{4})$,
\item single-source vertex connectivity has complexity $\Thetahat(n^{3})$, and
\item single-source-single-sink and global vertex connectivity have complexity $\Thetahat(n^{2})$.
\end{enumerate}
For graphs with general density, we obtain tight bounds of $\Thetahat(m^{2})$, $\Thetahat(m^{1.5})$, $\Thetahat(m)$, respectively, assuming Gomory-Hu trees for element connectivity can be computed in almost-linear time.

\begin{table}[]
\footnotesize

\caption{Upper bounds and lower bounds for connectivity problems}

\begin{tabular}{|c|c|c|c|c|}
\hline
 & Global & Single-Source & all-pairs & Steiner \\ \hline
\makecell{edge connectivity, \\ unweighted graphs}   & \makecell{$\tilde{O}(m)$\\\cite{Kar00}}  & $\hat{O}(m)$ \cite{AKL+22} & $\hat{O}(m)$ \cite{AKL+22} & $\hat{O}(m)$ \cite{LP20} \\ \hline
\multirow{2}*{\makecell{vertex connectivity, \\ unweighted graphs \\ with general density}}  &  &  \makecell{$\hat{\Omega}(m^{1.5})$ for comb. algo.,\\ \Cref{coro:SSVCDensity}} & \makecell{$\hat{\Omega}(m^2)$ for comb. algo.,\\ \Cref{coro:Density}}  & \makecell{$\hat{\Omega}(m^{2})$ for comb. algo.,\\Corollary of \Cref{thm:Steinerlowerbound}}  \\
\cline{3-5}
~ & \makecell{$\hat O(m)$\\\cite{LNP+21}} & \makecell{$\hat{O}(m^{1.5})$,\\assuming \Cref{conj:GHtree},\\\Cref{thm:SSVCupper}} & \makecell{$\hat{O}(m^{2})$,\\assuming \Cref{conj:GHtree},\\\Cref{thm:APVCupper}} & \makecell{$\hat{O}(m^{2})$,\\assuming \Cref{conj:GHtree},\\\Cref{thm:APVCupper}}\\
\cline{3-5}
~ & ~ & \makecell{$\hat{O}(m^{5/3})$,\\\Cref{thm:SSVCupper}} & \makecell{$\hat{O}(m^{11/5})$,\\\Cref{thm:APVCupper}} & \makecell{$\hat{O}(m^{11/5})$,\\\Cref{thm:APVCupper}}\\
\hline
\multirow{3}*{\makecell{vertex connectivity, \\ dense unweighted graphs\\$m=\Theta(n^{2})$}} & ~ & \makecell{$\hat{\Omega}(n^{3})$ for comb. algo.,\\
\Cref{coro:SSVC}}  & \makecell{$\hat{\Omega}(n^{4})$ for comb. algo.,\\\Cref{thm:APVClowerbound}}  & \makecell{$\hat{\Omega}(n^{4})$ for comb. algo.\\\Cref{thm:Steinerlowerbound}} \\ 
\cline{3-5}
~ & ~ & ~ & \makecell{$\hat{\Omega}(n^{\omega(1,2,1)})$ for all algo.,\\\Cref{remark:generalAPVC}}& ~\\
\cline{3-5}
~ & \makecell{$\hat{O}(n^2)$\\\cite{LNP+21}} & $\hat{O}(n^{3})$ trivially & $\hat{O}(n^{4})$ trivially & $\hat{O}(n^{4})$ trivially\\
\hline
\makecell{vertex connectivity,\\sparse weighted graphs,\\ $m=\tilde{O}(n)$} & ~ & ~ & $\hat{\Omega}(n^{3})$ \cite{AKT20} & ~ \\
\hline

\end{tabular}

\label{table}

\end{table}

\paragraph{Discussions and Open Problems.}
Our lower bounds for combinatorial algorithms are particularly relevant to the context of vertex connectivity since basically all algorithms for the problems are indeed combinatorial. 
There are a few exceptions \cite{linial1988rubber,abboud2018faster}, but these algorithms are still far from optimal (even cannot break our combinatorial lower bounds). 
It is a very interesting open problem whether one can bypass our combinatorial lower bounds using fast matrix multiplications, or show conditional lower bounds for general algorithms that match our combinatorial lower bounds.

\subsection{Technical Overview}
\label{sect:overview}
To prove \Cref{thm:1.1}, we are will reduce the 4-clique problem to the APVC problem. Previously, Abboud et al.~\cite{abboud2018faster} showed the hardness of all-pairs vertex connectivity in \emph{directed} graphs using the 4-clique problem, which inspired our reduction. However, the techniques are not strong enough to work on undirected graphs. In more details, the hard instance of \cite{abboud2018faster} is a directed acyclic graph with four layers, and so they only need to consider directed paths of length at most 3. In contrast, when we consider undirected graphs, paths connecting sources and sinks can be much more complex, which requires more advanced techniques and more careful arguments. Let us sketch our construction below.

Starting from a 4-clique instance $G$, it is helpful to consider its 4-partite version $G_{4p}$, which is simply constructed by duplicating $V(G)$ into 4 groups $A,B,C,D$ and copying $E(G)$ for each pair of different groups (see \Cref{def:4partite} for a formal definition). A natural way to answer the 4-clique problem on $G$ is then checking for each pair of adjacent $a\in A$ and $d\in D$, whether there exists an adjacent pair of vertices $b\in B$ and $c\in C$ that is adjacent to both $a$ and $d$ (call such $(b,c)$ a \emph{4-clique witness} of $(a,d)$). 

To correspond this to a vertex connectivity problem, for each pair of $a$ and $d$, consider a 4-layer graph $\hat{H}_{ad}$ defined as follows. Let $B_{a}$ denote the set of vertices in $B$ adjacent to $a$ (also define $B_{d}$, $C_{a}$ and $C_{d}$ similarly). Then $B_{a}\cap B_{d}$ (resp. $C_{a}\cap C_{d}$) are vertices in $B$ (resp. $C$) adjacent to both $a$ and $d$. The vertices of $\hat{H}_{ad}$ are
$V(\hat{H}_{ad})=\{a\}\cup(B_{a}\cap B_{d})\cup(C_{a}\cap C_{d})\cup\{d\}$, where the first layer (resp. the last layer) has only a single vertex $a$ (resp. $d$), and the second layer (resp. the third layer) has vertices $B_{a}\cap B_{d}$ (resp. $C_{a}\cap C_{d}$). The edge set of $\hat{H}_{ad}$ is $E(\hat{H}_{ad})=\{(a,b)\mid b\in B_{a}\cap B_{d}\}\cup E_{G_{4p}}(B_{a}\cap B_{d},C_{a}\cap C_{d})\cup \{(c,d)\mid c\in C_{a}\cap C_{d}\}$\footnote{Here $E_{G_{4p}}(B_{a}\cap B_{d},C_{a}\cap C_{d})\subseteq E(G_{4p})$ denote the set of edges connecting $B_{a}\cap B_{d}$ and $C_{a}\cap C_{d}$ in $G_{4p}$.}, which connects vertex $a$ (resp. $d$) to each second-layer (resp. third layer) vertex, and connects the second layer $B_{a}\cap B_{d}$ and the third layer $C_{a}\cap C_{d}$ using the same edges in $G_{4p}$. One can simply observe that a 4-clique witness $(b,c)$ of $(a,d)$ exists if and only if $\kappa_{\hat{H}_{ad}}(a,d)\geq 1$\footnote{In the overview, we use $\kappa_{H}(a,d)$ to denote the vertex connectivity between $a$ and $d$ in a graph $H$.}. To check the existence of 4-clique witnesses for all pairs $(a,d)$ simultaneously, our final APVC instance $H$ will be a combination of all $\hat{H}_{ad}$, and the main challenge is to make the combination compact.

We overcome this challenge by introducing two modules called the \emph{source-sink isolating gadgets} and the \emph{set-intersection filter}. Interestingly, our technique for proving time lower bounds is inspired by the space-lower bound techniques. More specifically, the source-sink isolating gadget is inspired by the construction in \cite{PSY22}.

The intuition behind these two modules is as follows. When considering $\kappa_{\hat{H}_{ad}}(a,d)$ of a specific pair of $a\in A$ and $d\in D$, we will somehow ``remove'' redundant vertices in $H\setminus \hat{H}_{ad}$ so that $\kappa_{\hat{H}_{ad}}(a,d)$ can be derived from $\kappa_{H}(a,d)$. The high-level idea to remove a redundant vertex $v$ in $H$ is to create a ``flow'' path from $a$ to $d$ through $v$ (e.g. a simple path made up of two edges $(a,v)$ and $(v,d)$). By a simple flow-cut argument, this path will enforce that $v$ appears in any $(a,d)$-vertex cut, while bringing some additive deviations to estimate $\kappa_{\hat{H}_{ad}}(a,d)$ as $\kappa_{H}(a,d)$. These modules apply this simple rule in a more general way. The source-sink isolating gadget will remove all vertices in $A$ and $D$ except $a$ and $d$, and the intersection patterns will generate $B_{a}\cap B_{d}$ and $C_{a}\cap C_{d}$ by removing other vertices in $B$ and $C$. The remaining graph will then be exactly $\hat{H}_{ad}$, and the additive deviations between $\kappa_{\hat{H}_{ad}}(a,d)$ and $\kappa_{H}(a,d)$ can be computed and subtracted easily.

The proof of \Cref{thm:SteinerVertexConn} extends the above idea to reduce a Steiner vertex connectivity problem to an edge-universal 4-clique problem. Consider an edge-universal 4-clique instance $G$. Based on the above construction of $H$, by creating more ``flow'' paths, we can guarantee that the additive deviations between $\kappa_{\hat{H}_{ad}}(a,d)$ and $\kappa_{H}(a,d)$ for all pair of $(a,d)$ will be the same, say a value $K$. Therefore, for each pair $(a,d)$, $\kappa_{H}(a,d)$ is either at least $K+1$ or equal to $K$, and there is no 4-clique containing $a$ and $d$ in the original graph $G$ if and only if the latter case happens. Finally, checking the Steiner vertex connectivity of terminal set $A\cup D$ suffices to answer the edge-universal 4-clique problem on $G$.

Towards the upper bound of the APVC problem on sparse graphs in \Cref{thm:SparseFastAlgo}, our algorithm uses the following scheme. Let the input graph be $G$ and let $k$ be a degree threshold to separate vertices into two parts, the high-degree part $V_{h}=\{v\in V(G)\mid \deg_{G}(v)>k\}$ and the low-degree part $V_{l}=\{v\in V(G)\mid \deg_{G}(v)\leq k\}$. For pairs $(u,v)$ such that $u,v\in V_{h}$, we can compute $\kappa_{G}(u,v)$ by simply calling max flows, which takes totally $O(m^{2}/k^{2})$ calls since $|V_{h}|=O(m/k)$. For other pairs $(u,v)$ with $u\in V_{l}$ or $v\in V_{l}$, there will be $\kappa_{G}(u,v)\leq k$, because the vertex connectivity of $u$ and $v$ is upper bounded by their degrees. The vertex connectivity oracle in \cite{PSY22} can exactly capture all-pairs vertex connectivity bounded by $k$, which takes $\tilde{O}(k^{2})$ black-box calls to Gomory-Hu trees for element connectivity whose construction time is currently $\hat{O}(mk)$. The whole algorithm takes $\hat{O}(m^{11/5})$ time by choosing a proper $k$, and the running time will be immediately improved to $\hat{O}(m^{2})$ if Gomory-Hu trees for element connectivity can be constructed in almost linear time.

\subsection{Organization}

We will start with some basic notations and introduce conjectures in \Cref{sect:prelim}. In \Cref{sect:APVClowerbound}, we will show the conditional lower bound of the APVC problem. In \Cref{sect:SteinerLowerBound}, we will extend the approach in \Cref{sect:APVClowerbound} to show the conditional lower bound of the Steiner vertex connectivity problem. In \Cref{sect:UpperBound}, we will show a simple APVC algorithm for graphs with general density.

\section{Preliminaries}
\label{sect:prelim}

In this paper, we use standard graph theoretic notations. For a graph $G$, we use $V(G)$ and $E(G)$ to denote its vertex set and edge set. For any vertex set $V'\subseteq V(G)$, we let $G[V']$ denote the subgraph induced by $V'$ and let $G\setminus V'$ be a short form of $G[V(G)\setminus V']$. For two graphs $G_{1}$ and $G_{2}$ which vertex sets $V(G_{1})$ and $V(G_{2})$ may intersect, we let $G_{1}\cup G_{2}$ denote the graph with vertex set $V(G_{1})\cup V(G_{2})$ and edge set $E(G_{1})\cup E(G_{2})$. For two subset of vertices $V_{1},V_{2}\subseteq V(G)$, we let $E_{G}(V_{1},V_{2})$ denote the set of edges directly connecting $V_{1}$ and $V_{2}$. For a vertex $v\in G$, we let $N_{G}(v)=\{u\mid (u,v)\in E(G)\}$ denote its neighbor set, and let $\bar{N}_{G}(v)=V(G) \setminus N_{G}(v)$ denote its non-neighbors.

\paragraph{Vertex connectivity.}
In a graph $G$, the vertex connectivity for two vertices $u,v\in V(G)$, denoted by $\kappa_G(u,v)$, is the maximum number of internally vertex-disjoint paths from $u$ to $v$. By Menger's theorem, $\kappa_{G}(u,v)$ is equal to the size of minimized subsets $C\subseteq(V(G)\setminus\{u,v\})\cup E(G)$ deleting which from $G$ will disconnect $u$ and $v$.

\paragraph{Vertex Connectivity problems.} In this paper, we will consider four vertex connectivity problems, i.e. the global, single-source, all-pairs, and Steiner vertex connectivity problems. The edge connectivity versions of these problems are analogous.

\begin{itemize}
    \item \textbf{The global vertex connectivity problem.} Given an undirected unweighted graph $G$, the global vertex connectivity problem (the global-VC problem) asks the global vertex connectivity, denoted by $\kappa_G$, where $\kappa_G=\min_{u,v\in G}\kappa_G(u,v)$.
    \item \textbf{The single-source vertex connectivity problem.} Given an undirected unweighted graph $G$ with a source vertex $s$, a single source vertex connectivity problem (the SSVC problem) asks $\kappa_{G}(s,v)$ for all other vertices $v\in G$.
    \item \textbf{The all-pairs vertex connectivity problem.} Given an undirected unweighted graph $G$, the all-pairs vertex connectivity problem (the APVC problem) asks $\kappa_{G}(u,v)$ for all pairs of $u,v\in G$.
    \item \textbf{The Steiner vertex connectivity problem.} Given an undirected unweighted graph $G$ with a terminal set $T\subseteq V(G)$, the Steiner vertex connectivity problem (the Steiner-VC problem) asks the Steiner vertex connectivity of $T$, denoted by $\kappa_G(T)$, where $\kappa_G(T)=\min_{u,v\in T}\kappa_G(u,v)$.
\end{itemize}

\paragraph{The 4-clique conjecture.}
Given an $n$-vertex undirected graph $G$, the $k$-clique problem is to decide whether there is a clique with $k$ vertices in $G$. The $k$-clique problem can be solved in $O(n^{k})$ time trivially, and a more efficient combinatorial algorithm takes running time $O(n^{k}/\log^{k}n)$ \cite{Vas09}. The popular $k$-clique conjecture (see e.g.~\cite{BGL17}) suggests that there is no combinatorial algorithm for the $k$-clique problem with $O(n^{k-\epsilon})$ running time for any constant $\epsilon>0$. In \Cref{sect:APVClowerbound}, we will use this conjecture in the case $k=4$.

\begin{conjecture}[4-clique conjecture]
There is no combinatorial algorithm that solves the 4-clique problem for $n$-vertex graphs in $O(n^{4-\epsilon})$ time for any constant $\epsilon>0$.
\label{conj:4clique}
\end{conjecture}

For each 4-clique instance $G$, it is equivalent to consider its 4-partite form $G_{4p}$ as defined below. Note that a 4-clique in a 4-partite graph should contain exactly one vertex from each group, and the original graph has a 4-clique if and only if the 4-partite graph $G_{4p}$ has a 4-clique.

\begin{definition}[4-partite graph $G_{4p}$]
Given an undirected graph $G$, the 4-partite graph $G_{4p}$ of $G$ has vertex set $V(G_{4p})=\{v_{A},v_{B},v_{C},v_{D}\mid v\in V(G)\}$, and we let $A=\{v_{A}\mid v\in V(G)\},B=\{v_{B}\mid v\in V(G)\},C=\{v_{C}\mid v\in V(G)\},D=\{v_{D}\mid v\in V(G)\}$ be four groups partitioning $V(G_{4p})$. The edge set $E(G_{4p})=\{(u_{X},v_{Y})\mid (u,v)\in E(G),X,Y\in\{A,B,C,D\},X\neq Y\}$.
\label{def:4partite}
\end{definition}

\paragraph{The edge-universal 4-clique conjecture.}
We consider a variant of the 4-clique problem, called the \emph{edge-universal 4-clique problem}. Given an undirected graph $G$ and a subset of \emph{demand} edges $E_{\dem}\subseteq E(G)$, this problem asks if every edge in $E_{\dem}$ is contained by a 4-clique. We suggest the following conjecture on this problem.

\begin{conjecture}[Edge-universal 4-clique conjecture]
There is no combinatorial algorithm that answers the all edge 4-clique problem for $n$-vertex graphs in $O(n^{4-\epsilon})$ time for any constant $\epsilon>0$.
\label{conj:AllEdge4Clique}
\end{conjecture}

We note our formulation of the edge-universal 4-clique problem is at least as hard as the problem of checking if every edge is contained in some 4-clique by fixing $E_{\dem} = E(G)$. We choose to present this formulation that allows any $E_{\dem} \subseteq E(G)$ because it only strengthens our hardness result and shows the flexibility of our reduction.

The difference between the edge-universal 4-clique conjecture vs. the 4-clique conjecture is that we switch the quantifier in the definition of the problems. This allows us to obtain new hardness results. 
This strategy for proving conditional lower bounds has been studied and done several times in the literature of fine-grained complexity \cite{gao2018completeness,abboud2016approximation,abboud2022scheduling}.
For example, the relationship between the edge-universal 4-clique problem vs.~the 4-clique problem is the same as the relationship between the well-known \emph{orthogonal vector} problem vs.~the \emph{hitting set} problem introduced in \cite{abboud2016approximation}, and \emph{SETH} vs. \emph{quantified SETH} introduced in \cite{abboud2022scheduling}.

\paragraph{Gomory-Hu trees for element connectivity}
Gomory-Hu trees are cut-equivalent trees introduced by Gomory and Hu \cite{GH61} to \emph{capture} all-pairs edge connectivity in weighted graphs. More precisely, given a Gomory-Hu tree, the edge connectivity of any given pair of vertices can be queried in nearly constant time. Very recently, a breakthrough by \cite{AKL+22} showed that a Gomory-Hu tree can be constructed in $\tilde{O}(n^{2})$ time for a weighted graph and $\hat{O}(m)$ time for an unweighted graph. For vertex connectivity, it has been shown by \cite{Ben95} that there are no such cut-equivalent trees. See also \cite{PSY22} for a more general space lower bound forbidding the existence of such trees for vertex connectivity.

Element connectivity is the notion of connectivity that is related to  vertex connectivity, and yet Gomory-Hu trees have been shown to exist for element connectivity (see e.g. \cite{CRX15}). Given an undirected graph $G$ and a terminal set $U\subseteq V(G)$, the element connectivity for two vertices $u,v\in U$, denoted by $\kappa'_{G, U}(u,v)$, is the size of minimized set $C\subseteq E(G)\cup(V(G)\setminus U)$ whose removal will disconnect $u$ and $v$. An element connectivity Gomory-Hu tree for the graph $G$ and terminal set $U$ will capture $\kappa'_{G,U}(u,v)$ for all pairs of $u,v\in U$.

In \cite{PSY22}, they consider a variant called $k$-Gomory-Hu tree for element connectivity, which given an additional parameter $k$, will capture the value $\min\{\kappa'_{G,U}(u,v),k\}$ for all $u,v\in U$ (namely we can query $\min\{\kappa'_{G,U}(u,v),k\}$ for each given $u,v\in U$ in nearly constant time). By generalizing the $(1+\epsilon)$-approximate Gomory-Hu tree algorithm by \cite{LP21} to the element connectivity setting, the following result was obtained by \cite{PSY22}.

\begin{theorem}
Given an $n$-vertex $m$-edge undirected unweighted graph $G$, a terminal set $U\subseteq V(G)$ and a parameter $k$, there is a randomized algorithm to construct a $k$-Gomory-Hu tree for element connectivity in $\hat{O}(mk)$ time with high probability.

\label{thm:GHtree}
\end{theorem}

Given the similarity of Gomory-Hu trees for edge connectivity and element connectivity, and the recent breakthrough of $\hat{O}(m)$-time construction algorithm for edge connectivity Gomroy-Hu tree, it seems reasonable to conjecture that element connectivity Gomory-Hu tree can also be constructed in $\hat{O}(m)$ time. 

\begin{conjecture}

Given an $n$-vertex $m$-edge undirected unweighted graph $G$ and a terminal set $U\subseteq V(G)$, an element connectivity Gomory-Hu tree can be constructed in $\hat{O}(m)$ time.

\label{conj:GHtree}

\end{conjecture}

We leave this conjecture as a very interesting open problem.

\section{The Lower Bound for the APVC Problem}
\label{sect:APVClowerbound}

In this section, we will prove \Cref{thm:APVClowerbound}, a conditional lower bound of the APVC problem in undirected unweighted graphs conditioning on the 4-clique conjecture. Concretely, we will show a reduction from the 4-clique problem to the APVC problem.

\begin{theorem}
Assuming \Cref{conj:4clique}, for $n$-vertex undirected unweighted graphs, there is no combinatorial algorithm that solves the APVC problem in $O(n^{4-\epsilon})$ time for any constant $\epsilon > 0$. 
\label{thm:APVClowerbound}
\end{theorem}

Given an $n$-vertex 4-clique instance $G$, let $G_{4p}$ be the corresponding 4-partite graph defined in \Cref{def:4partite} (where $V(G_{4p})$ is partitioned into 4 groups $A,B,C,D$). We start with some notations. For each vertex $a\in A$, we use $B_{a}=\{b\in B|(a,b)\in E(G_{4p})\}$ to denote the neighbors of $a$ in $B$ and let $\bar{B}_{a}=B\setminus B_{a}$. Analogously, $C_{a}$ is the set of vertices in $C$ adjacent to $a$ and $\bar{C}_{a}=C\setminus C_{a}$. For each $d\in D$, we define $B_{d},\bar{B}_{d},C_{d},\bar{C}_{d}$ in a similar way. 

As discussed in \Cref{sect:overview}, our 4-clique instance $H$ will be constructed using two kinds of modules, the source-sink isolating gadgets and the set-intersection filters, which will be introduced in \Cref{sect:CleanupGadget} and \Cref{sect:IntersectionPattern} respectively. After that, the final construction of $H$ and the proof of \Cref{thm:APVClowerbound} will be completed in \Cref{sect:APVCreduction}.

\subsection{The Source-Sink Isolating Gadget}
\label{sect:CleanupGadget}

We first introduce the \emph{source-sink isolating gadget}. Basically, for an undirected graph $R$ and two disjoint groups of vertices $X,Y\subseteq V(R)$, a source-sink isolating gadget $Q(X,Y)$ (or just $Q$ for short) is a graph on vertices $X\cup Y$ with  additional vertices outside $R$. Its formal guarantee is as follows.

\begin{lemma}
Given an undirected graph $R$ and two disjoint groups of vertices $X,Y\subseteq V(R)$, there is a graph $Q$ with $V(Q)\cap V(R)=X\cup Y$ and $|V(Q)|=O(|X|+|Y|)$ such that for any $x\in X,y\in Y$ with $(x,y)\notin E(R)$,
\[\kappa_{R\cup Q}(x,y)=\kappa_{R_{xy}}(x,y)+|X|+|Y|,\]
where $R_{xy}=R\setminus((X\cup Y)\setminus\{x,y\})$. Such a graph $Q$ is called a source-sink isolating gadget, and moreover, the construction of $Q$ is independent from $R$.
\label{lemma:CleanupGadget}
\end{lemma}

The reason we call the graph $Q$  a source-sink isolating gadget is that by adding $Q$ into the input graph $R$ the vertex connectivity between any pair of source $x\in X$ and sink $y \in Y$ in $R\cup Q$, i.e., $\kappa_{R\cup Q}(x,y)$, can be derived from their connectivity in $R_{xy}$, i.e., $\kappa_{R_{xy}}(x,y)$. But the graph $R_{xy}$, as defined in  \Cref{lemma:CleanupGadget}, is just the graph $R$ after removing all source and sink vertices in $X$ and $Y$ except $x$ and $y$. That is, the gadget  ``isolates'' the pair $x$ and $y$ from the rest. We will use this gadget in \Cref{sect:APVCreduction}.

\begin{proof}
We construct $Q$ in the following way. We create duplicated sets $X_{1}, X_{2}$ of $X$, and also $Y_{1}, Y_{2}$ of $Y$. For each vertex $x\in X$, we let $\hat{x}_{1}\in X_{1}$ and $\hat{x}_{2}\in X_{2}$ denote copies of $x$ in $X_{1}$ and $X_{2}$ respectively if there is no other specification (for each $y\in Y$, $\hat{y}_{1},\hat{y}_{2}$ are defined similarly). The vertex set of $Q$ is $V(Q)=X\cup X_{1}\cup X_{2}\cup Y\cup Y_{1}\cup Y_{2}$, and the edge set is
\begin{align*}
    E(Q)=&\{(x,\hat{x}_{1})\mid x\in X\}\cup\{(x_{1},y)\mid x_{1}\in X_{1},y\in Y\}\cup\\
    &\{(y,\hat{y}_{1})\mid y\in Y\}\cup\{(y_{1},x)\mid y_{1}\in Y_{1},x\in X\}\cup\\
    &\{(x,x_{2})\mid x\in X,x_{2}\in X_{2}\}\cup\\
    &\{(y,y_{2})\mid y\in Y,y_{2}\in Y_{2}\}.
\end{align*}
The construction of $Q$ is illustrated in \Cref{CG}.

\begin{figure}[ht]
    \centering    \includegraphics[width=\textwidth]{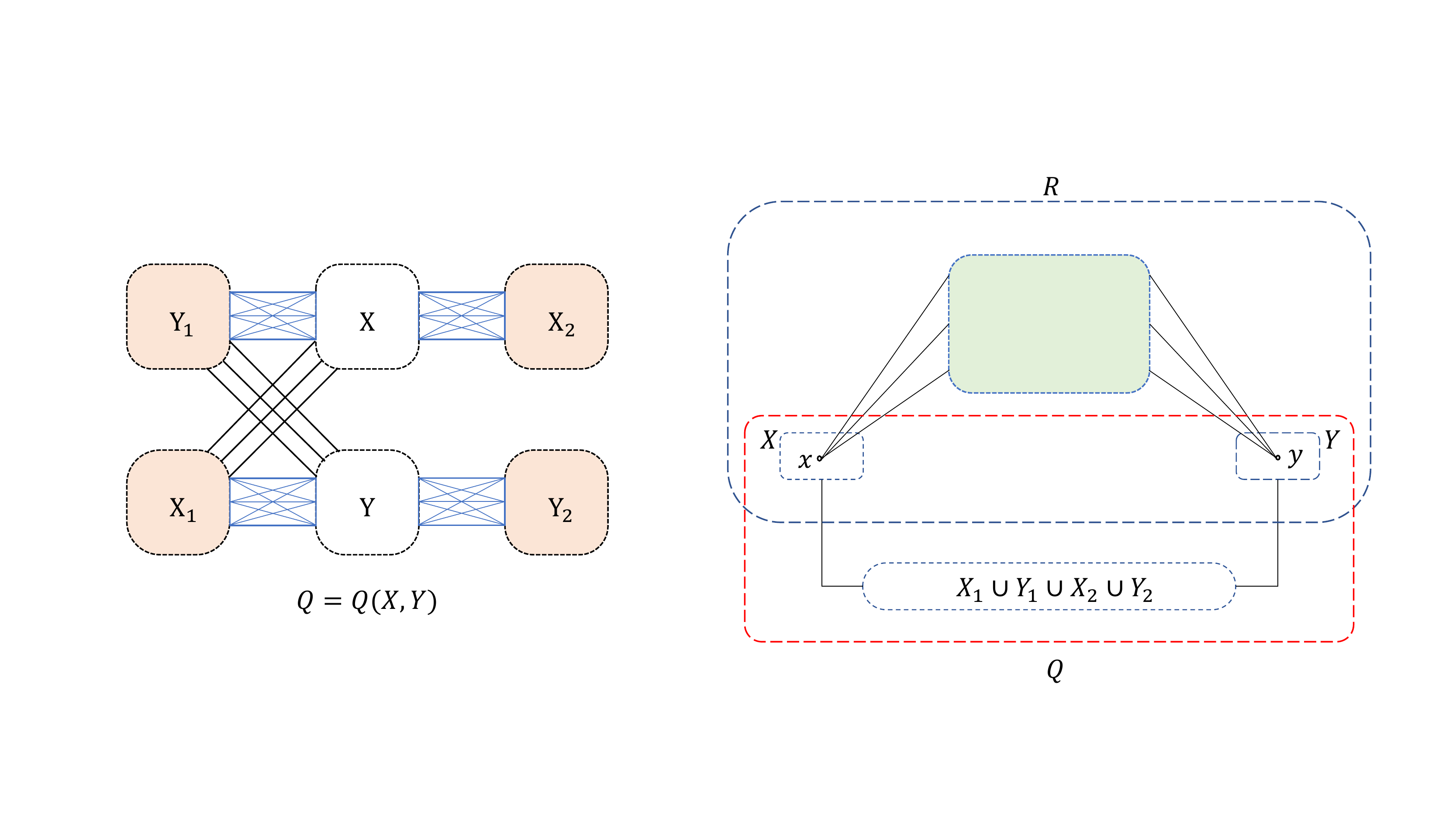}
    \caption{The source-sink isolating gadget $Q$ and the whole graph $R\cup Q$. There are bipartite cliques between $X$ and $Y_1,X_2$, as well as $Y$ and $X_1,Y_2$, and there are perfect matchings between $X$ and $X_1$, $Y$ and $Y_1$.}
    \label{CG}
\end{figure}

Fixing some $x\in X$ and $y\in Y$, we let $\kappa=\kappa_{R_{xy}}(x,y)$ for short. We first show $\kappa_{R\cup Q}(x,y)\geq \kappa+|X|+|Y|$. From the flow view of vertex connectivity, there are $\kappa$ internal vertex-disjoint paths from $x$ to $y$ in $R_{xy}$. Combining \Cref{claim:PathNumberQ} and $V(Q)\cap V(R_{xy})=\{x,y\}$, there are $\kappa+|X|+|Y|$ internal vertex-disjoint paths in $R\cup Q$. Therefore, $\kappa_{R\cup Q}(x,y)\geq \kappa+|X|+|Y|$.

\begin{claim}
There are $|X|+|Y|$ internal vertex-disjoint paths from $x$ to $y$ in $Q$. 
\label{claim:PathNumberQ}
\end{claim}
\begin{proof}
The first path is $x\to \hat{x}_{1}\to y$. Each of the next $|X|-1$ paths corresponds to each $x'\in X$ s.t. $x'\neq x$, the path $x\to \hat{x}'_{2}\to x'\to \hat{x}'_{1}\to y$ concretely, where $\hat{x}'_{1}$ and $\hat{x}'_{2}$ are copies of $x'$ in $X_{1}$ and $X_{2}$. Symmetrically, there is a path $x\to y_{1}\to y$ and $|Y|-1$ paths, each of which corresponds to each $y'\in Y$ s.t. $y'\neq y$ (namely the path $x\to \hat{y}'_{1}\to y'\to \hat{y}'_{2}\to y$, where $\hat{y}'_{1}$ and $\hat{y}'_{2}$ are copies of $y'$ in $Y_{1}$ and $Y_{2}$). Observe that these $|X|+|Y|$ paths are internal vertex-disjoint.
\end{proof}

We then argue from the cut view that $\kappa_{R\cup Q}(x,y)\leq \kappa+|X|+|Y|$. Consider the vertex set $S_{Q}=\{x'\mid x'\in X,x'\neq x\}\cup\{y'\mid y'\in Y,y'\neq y\}\cup\{\hat{x}_{1},\hat{y}_{1}\}$. After removing $S_{Q}$ from $R\cup Q$, observe that vertices in both $R$ and $Q$ are only $x$ and $y$, so a simple path from $x$ to $y$ in graph $(R\cup Q)\setminus S_{Q}$ will be totally inside subgraphs $Q\setminus S_{Q}$ or $R\setminus S_{Q}$. Note that $x$ and $y$ are disconnected in $Q\setminus S_{Q}$. Moreover, subgraph $R\setminus S_{Q}$ is exactly $R_{xy}$, so removing $\kappa$ vertices can disconnect $x$ and $y$ in $R\setminus S_{Q}$. In conclusion, in graph $R\cup Q$, $x$ and $y$ can be disconnected by removing $|S_{Q}|+\kappa$ vertices, so $\kappa_{R\cup Q}(x,y)\leq\kappa + |X| + |Y|$.

Finally, the size of $Q$ follows directly from the construction.

\end{proof}

\subsection{The Set-Intersection Filter}
\label{sect:IntersectionPattern}
We now introduce the set-intersection filter.
For each $a\in A$, $d\in D$, the set-intersection filter $P_{ad}^{B}$ is a subgraph of the final $H$, which will ``filter'' the intersection $B_{a}\cap B_{d}$ from the whole set $B$ as \Cref{lemma:IntersectionPatternB} shows. It is constructed as follows. Let $V(P_{ad}^{B})=\{a\}\cup B\cup B'\cup \{d\}$, where $B'$ duplicates vertices in $B$. For each vertex $b\in B$, we use $\hat{b}'$ to denote its copy in $B'$, and for each (non-)neighbor sets $B_{a},\bar{B}_{a},B_{d},\bar{B}_{d}\subseteq B$, we use $B'_{a},\bar{B}'_{a},B'_{d},\bar{B}'_{d}\subseteq B'$ to denote their counterparts respectively. The edge set of $P_{ad}^{B}$ is constructed by
\begin{align*}
E(P_{ad}^{B})=&\{(a,b)\mid b\in B\}\cup\{(b,d)\mid b\in \bar{B}_{d}\}\cup\\
&\{(a,\hat{b}')\mid b\in B_{a}\}\cup\{(b',d)\mid b'\in B'\}\cup\\
&\{(b,\hat{b}')\mid b\in B\}.
\end{align*}
See \Cref{IP} for an illustration of $P^{B}_{ad}$.

\begin{figure}[ht]
    \centering
    \includegraphics[width=\textwidth]{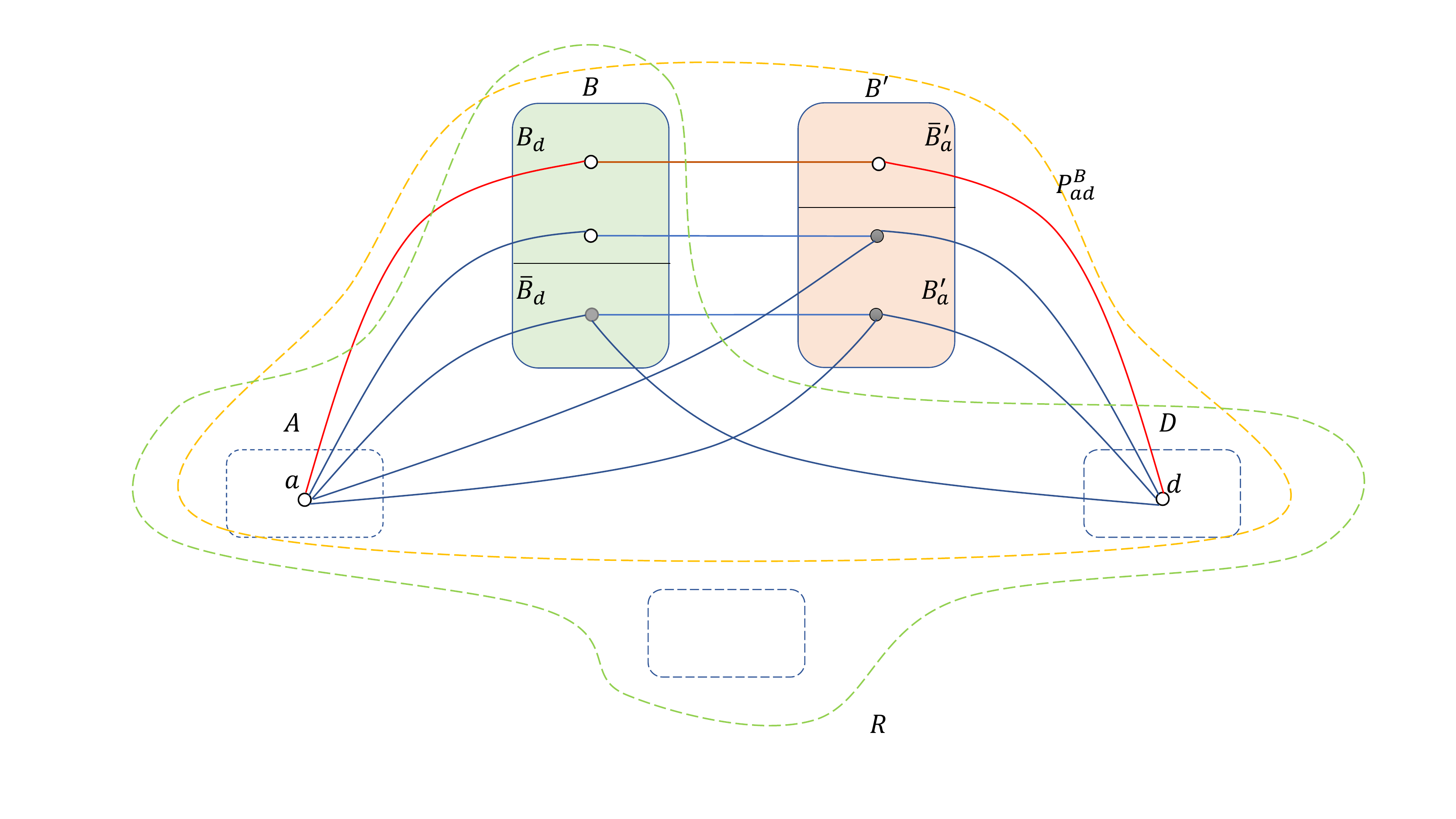}
    \caption{The set-intersection filter. The sets $P_{ad}^{B}$ and $R$ are the areas surrounded by dotted lines.}
    \label{IP}
\end{figure}

The construction of $E(P^{B}_{ad})$ can be interpreted in the following intuitive way. 
\begin{itemize}
\item First, the edges $\{(a,b)\mid b\in B\}$ and $\{(b,d)\mid b\in\bar{B}_{d}\}$ create vertex-disjoint paths of the format $a\to b\to d$ for all $b\in\bar{B}_{d}$, which implies $\bar{B}_{d}$ will be cut from $B$ in every vertex cut of $(a,d)$. 
\item Second, the edges $\{(a,\hat{b}')\mid b\in B_{a}\}$ and $\{(b',d)\mid b'\in B'\}$ create vertex-disjoint paths of the format $a\to \hat{b}'\to d$ for all $b\in B_{a}$, which analogously implies that $B'_{a}$ will be cut from $B'$ in every $(a,d)$-vertex cut.

\item Third, for every $(a,d)$-vertex cut, after $\bar{B}_{d}$ and $B'_{a}$ are cut from $B$ and $B'$ respectively from the above discussion, either $b$ or $\hat{b}'$ should be cut for all $b\in B_{d}\cap \bar{B}_{a}$. The reason is that the edges $\{(b,\hat{b}')\mid b\in B\}$ form a matching between $B$ and $B'$, which will create vertex-disjoint paths of the format $a\to b\to \hat{b}'\to d$ for all $b\in B_{d}\cap\bar{B}_{a}$. 

\end{itemize}

Therefore, suppose that in the third step we choose to cut vertex $b$ of all $b\in B_{d}\cap \bar{B}_{a}$. (In the formal proof of \Cref{lemma:IntersectionPatternB}, we will see that, cutting $b$ rather than $\hat{b}'$ for all $b\in B_{d}\cap \bar{B}_{a}$ is always a better choice when considering vertex min cut between $a$ and $d$.)
Now $\bar{B}_{d}$ is cut in the first step and $B_{d}\cap \bar{B}_{a}$ is cut in the third step, so
vertices in $B$ that survive are $B\setminus\bar{B}_{d}\setminus(B_{d}\cap\bar{B}_{a})=B_{a}\cap B_{d}$. Therefore, the set-intersection filter indeed obtain $B_{a}\cap B_{d}$ as desired. 

\begin{lemma}
For each $a\in A,d\in D$, the set-intersection filter $P^{B}_{ad}$ has the following property. Let $R$ be an undirected graph such that $V(R)\cap V(P^{B}_{ad})=\{a\}\cup B\cup\{d\}$, then
\[
    \kappa_{R\cup P^{B}_{ad}}(a,d)=\kappa_{R^{B}_{ad}}(a,d)+|\bar{B}_{d}|+|B_{a}|+|B_{d}\cap\bar{B}_{a}|,
\]
where $R^{B}_{ad}=(R\cup P^{B}_{ad})\setminus((B\setminus(B_{a}\cap B_{d}))\cup B')$ equivalent to $(R\setminus (B\setminus(B_{a}\cap B_{d})))\cup \{(a,b)\mid b\in B_{a}\cap B_{d}\}$, that is, the graph starting from $R$, removing non $B_{a}\cap B_{d}$ vertices and then adding edges connecting $a$ and each $b\in B_{a}\cap B_{d}$.
\label{lemma:IntersectionPatternB}
\end{lemma}

\begin{proof}
Let $\kappa=\kappa_{R^{B}_{ad}}(a,d)$ for short. We first show $\kappa_{R\cup P^{B}_{ad}}(a,d)\geq \kappa+|\bar{B}_{d}|+|B_{a}|+|B_{d}\cap \bar{B}_{a}|$. There are $\kappa$ internal vertex-disjoint paths from $a$ to $d$ in $R^{B}_{ad}$. Because $V(P^{B}_{ad})\cap V(R^{B}_{ad})=(B_{a}\cap B_{d})\cup\{a,d\}$, the paths from $a$ to $d$ in $P^{B}_{ad}\setminus(B_{a}\cap B_{d})$ are internal disjoint with those in $R^{B}_{ad}$, and there are $|\bar{B}_{d}|+|B_{a}|+|B_{d}\cap \bar{B}_{a}|$ of them by \Cref{claim:PathNumberP}. Therefore, we have $\kappa+|\bar{B}_{d}|+|B_{a}|+|B_{d}\cap\bar{B}_{a}|$ internal vertex-disjoint paths from $a$ to $d$ in $R\cup P^{B}_{ad}$.

\begin{claim}
There are $|\bar{B}_{d}|+|B_{a}|+|B_{d}\cap \bar{B}_{a}|$ internal vertex-disjoint paths from $a$ to $d$ in $P^{B}_{ad}\setminus(B_{a}\cap B_{d})$.
\label{claim:PathNumberP}
\end{claim}
\begin{proof}
The first $|\bar{B}_{d}|$ paths correspond to vertices $b\in\bar{B}_{d}$, each of which has a path $a\to b\to d$. The next $|B_{a}|$ paths correspond to vertices $b\in B_{a}$, each of which has a path $a\to \hat{b}'\to d$. The last $|B_{d}\cap \bar{B}_{a}|$ paths correspond to vertices $b\in B_{d}\cap \bar{B}_{a}$, each of which has a path $a\to b\to \hat{b}'\to d$.
\end{proof}

We then complete the proof by showing $\kappa_{R\cup P^{B}_{ad}}(a,d)\leq \kappa + |\bar{B}_{d}|+|B_{a}|+|B_{d}\cap \bar{B}_{a}|$. Let $S_{P}=\{b\in B\mid b\in \bar{B}_{d}\}\cup\{\hat{b}'\in B'\mid b\in B_{a}\}\cup\{b\in B\mid b\in B_{d}\cap \bar{B}_{a}\}$ be a vertex cut of $(a,d)$ in $P^{B}_{ad}$. Observe that in $(R\cup P^{B}_{ad})\setminus S_{P}$, a path from vertex $a$ to a vertex $b'\in B'\setminus S_{P}$ must go through vertex $d$, because each $b'\in B'\setminus S_{P}$ only connects to $d$ after removing $S_{P}$. Therefore, it is safe to ignore $B'\setminus S_{P}$ when considering the vertex connectivity between $a$ and $d$ in graph $(R\cup P^{B}_{ad})\setminus S_{P}$, i.e. 
\[
\kappa_{(R\cup P^{B}_{ad})\setminus S_{P}}(a,d)=\kappa_{(R\cup P^{B}_{ad})\setminus (S_{P}\cup B')}(a,d)=\kappa_{R^{B}_{ad}}(a,d)=\kappa,
\]
which means by further removing $\kappa$ vertices, we can disconnect $a$ and $d$ in $(R\cup P_{ad}^{B})\setminus S_{P}$.

In conclusion, we can remove $|S_{P}|+\kappa$ vertices to disconnect $a$ and $d$ in $R\cup P^{B}_{ad}$, which implies $\kappa_{R\cup P^{B}_{ad}}(a,d)\leq \kappa + |\bar{B}_{d}|+|B_{a}|+|B_{d}\cap \bar{B}_{a}|$.

\end{proof}

For each $a\in A,d\in D$, we also define the set-intersection filter $P^{C}_{ad}$ similarly but symmetrically. Let $V(P^{C}_{ad})=\{a\}\cup C'\cup C\cup \{d\}$, where $C'$ duplicates vertices in $C$. For each $c\in C$, let $\hat{c}'$ denote its copy in $C'$. For each (non-)neighbor sets $C_{a},\bar{C}_{a},C_{d},\bar{C}_d\subseteq C$, let $C'_{a},\bar{C}'_{a},C'_{d},\bar{C}'_d\subseteq C'$ denote their counterparts respectively.
The edge set is
\begin{align*}
    E(P^{C}_{ad})=&\{(a,c)\mid c\in \bar{C}_{a}\}\cup\{(c,d)\mid c\in C\}\cup\\
    &\{(a,c')\mid c'\in C'\}\cup\{(\hat{c}',d)\mid c\in C_{d}\}\cup\\
    &\{(\hat{c}',c)\mid c\in C\}.
\end{align*}
The pattern $P^{C}_{ad}$ will also have similar properties as shown below.

\begin{lemma}
For each $a\in A,d\in D$, the set-intersection filter $P^{C}_{ad}$ has the following property. Let $R$ be an undirected graph such that $V(R)\cap V(P^{C}_{ad})=\{a\}\cup C\cup\{d\}$, then
\[
    \kappa_{R\cup P^{C}_{ad}}(a,d)=\kappa_{R^{C}_{ad}}(a,d)+|C_{d}|+|\bar{C}_{a}|+|C_{a}\cap\bar{C}_{d}|,
\]
where $R^{C}_{ad}=(R\cup P^{C}_{ad})\setminus((C\setminus(C_{a}\cap C_{d}))\cup C')$ equivalent to $(R\setminus(C\setminus(C_{a}\cap C_{d})))\cup\{(c,d)\mid c\in C_{a}\cap C_{d}\}$, that is, the graph starting from $R$, removing non $C_{a}\cap C_{d}$ vertices and then adding edges connecting $d$ and each $c\in C_{a}\cap C_{d}$.
\label{lemma:IntersectionPatternC}
\end{lemma}

\subsection{The Final Construction of the APVC Instance}
\label{sect:APVCreduction}

We are now ready to construct the final APVC instance $H$. For each $a\in A,d\in D$, we first construct a graph $H_{ad}$ as follows. Let $P^{B}_{ad}$ and $P^{C}_{ad}$ be the set-intersection filters defined in \Cref{sect:IntersectionPattern}. Then the graph $H_{ad}$ will be defined by
\[
H_{ad}=P^{B}_{ad}\cup P^{C}_{ad}\cup G_{4p}[B\cup C],
\]
which is the union of two set-intersection filters with edges in the graph $G_{4p}$ connecting $B$ and $C$. The final graph then will be constructed by
\[
H=\bigcup_{a\in A,d\in D}H_{ad}\cup Q(A,D),
\]
where $Q(A,D)$ is the source-sink isolating gadget from \Cref{lemma:CleanupGadget} given $R=\bigcup_{a\in A,d\in D}H_{ad}$ and the sets $A,D$. See \Cref{APVC} below for an illustration.

\begin{figure}[ht]
    \centering
    \includegraphics[width=1\textwidth]{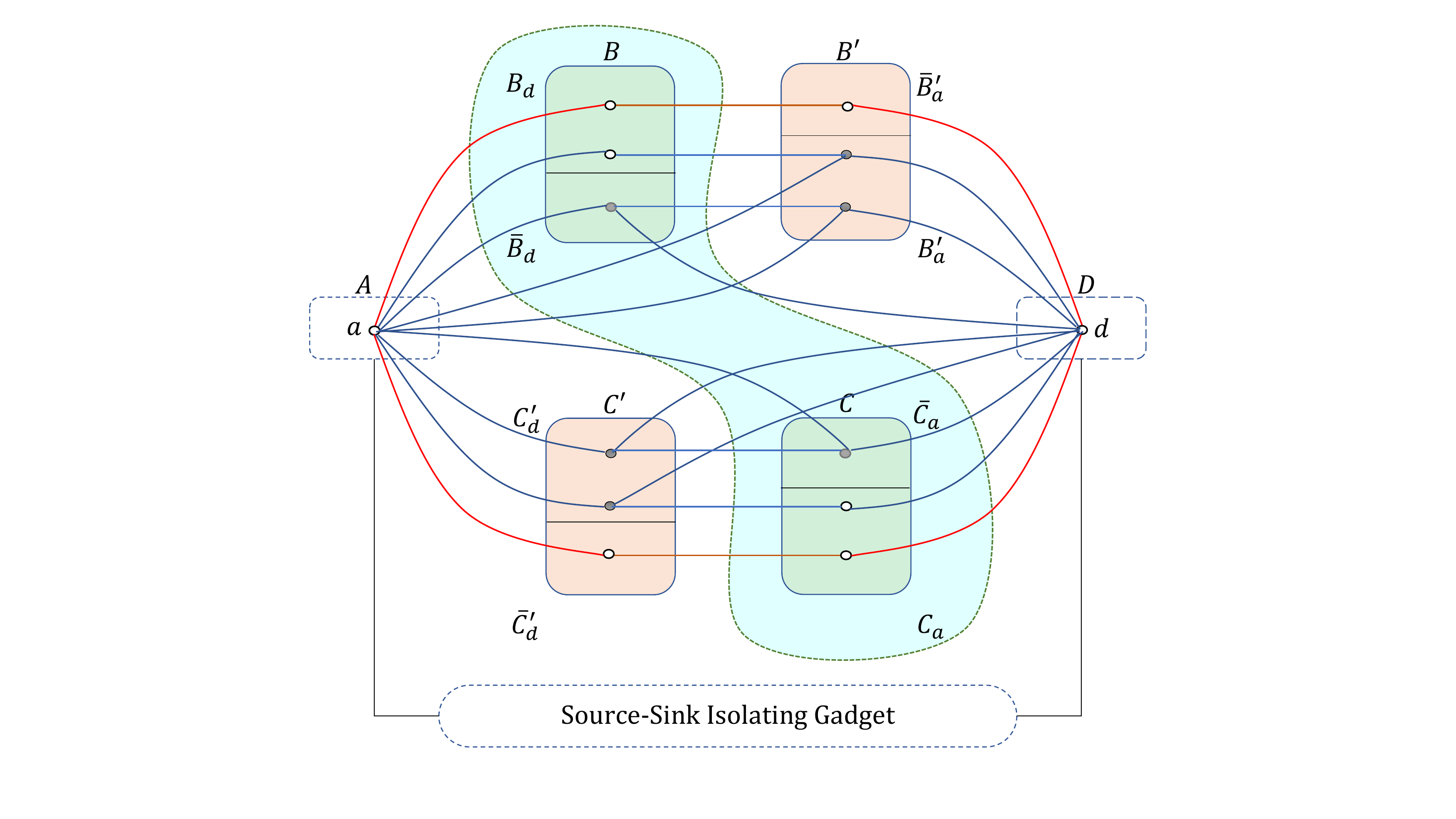}
    \caption{The graph $H$. The edges connecting $B$ and $C$ (i.e. $G_{4p}[B\cup C]$) are omitted in the highlighted part.}
    
    \label{APVC}
\end{figure}

\begin{lemma}
For each $a\in A$ and $d\in D$ of $G_{4p}$, we have
\begin{equation}
    \kappa_{H}(a,d)\geq 4n+|N_{G}(a)\cap \bar{N}_{G}(d)| + |\bar{N}_{G}(a)\cap N_{G}(d)|.
\label{eq:ConnLB}
\end{equation}
Furthermore, without ambiguity let $a$ and $d$ also denote their original vertices in $G$. Then there is a 4-clique in $G$ containing $a$ and $d$ if and only if $a$ and $d$ are adjacent in $G_{4p}$ and
\begin{equation}
    \kappa_{H}(a,d)\geq 4n+|N_{G}(a)\cap \bar{N}_{G}(d)| + |\bar{N}_{G}(a)\cap N_{G}(d)|+1.
\label{eq:Conn=4clique}
\end{equation}
\label{lemma:reduction}
\end{lemma}

\begin{proof}
Fixing some $a\in A,d\in D$, we first apply \Cref{lemma:CleanupGadget} on $H$ and vertex sets $A,D$, which gives $\kappa_{H}(a,d)=\kappa_{R_{ad}}(a,d)+|A|+|D|$,
where $R_{ad}=(\bigcup_{a\in A,d\in D}H_{ad})\setminus((A\cup D)\setminus\{a,d\})$. In fact, the graph $R_{ad}$ is exactly $H_{ad}$, because the induced subgraphs $H_{ad}[B\cup B'\cup C\cup C']$ are the same for all $a\in A,d\in D$ by construction. Therefore, we have 
\begin{equation}
\kappa_{H}(a,d)=\kappa_{H_{ad}}(a,d)+|A|+|D|.
\label{eq:1}
\end{equation}

Recall that $H_{ad}=P^{B}_{ad}\cup P^{C}_{ad}\cup G_{4p}[B\cup C]$. Let $R_{1}=P^{C}_{ad}\cup G_{4p}[B\cup C]$. We apply \Cref{lemma:IntersectionPatternB} on $P^{B}_{ad}$ and $R_{1}$, which gives
\begin{equation}
\kappa_{H_{ad}}(a,d)=\kappa_{R_{1}\cup P^{B}_{ad}}(a,d)=\kappa_{R'_{1}}(a,d)+|\bar{B}_{d}|+|B_{a}|+|B_{d}\cap \bar{B}_{a}|,
\label{eq:2}
\end{equation}
where
\begin{align*}
R'_{1}&=(R_{1}\setminus(B\setminus(B_{a}\cap B_{d})))\cup\{(a,b)\mid b\in B_{a}\cap B_{d}\}\\
&=P^{C}_{ad}\cup G_{4p}[(B_{a}\cap B_{d})\cup C]\cup\{(a,b)\mid b\in B_{a}\cap B_{d}\}.
\end{align*}

Let $R_{2}=G_{4p}[(B_{a}\cap B_{d})\cup C]\cup\{(a,b)\mid b\in B_{a}\cap B_{d}\}$. We apply \Cref{lemma:IntersectionPatternC} on $P^{C}_{ad}$ and $R_{2}$, which gives
\begin{equation}
\kappa_{R'_{1}}(a,d)=\kappa_{R_{2}\cup P^{C}_{ad}}(a,d)=\kappa_{R'_{2}}(a,d)+|C_{d}|+|\bar{C}_{a}|+|C_{a}\cap\bar{C}_{d}|,
\label{eq:3}
\end{equation}
where
\begin{align*}
R'_{2}=&(R_{2}\setminus(C\setminus(C_{a}\cap C_{d})))\cup \{(c,d)\mid c\in C_{a}\cap C_{d}\}\\
=&G_{4p}[(B_{a}\cap B_{d})\cup(C_{a}\cap C_{d})]\cup\{(a,b)\mid b\in B_{a}\cap B_{d}\} \cup\\ & \{(c,d)\mid c\in C_{a}\cap C_{d}\}.
\end{align*}

We use $\hat{H}_{ad}$ to denote $R'_{2}$ in the remaining proof and note that it is equivalent to the definition of $\hat{H}_{ad}$ in \Cref{sect:overview}. Because $\kappa_{\hat{H}_{ad}}(a,d)\geq 0$, combining \Cref{eq:1,eq:2,eq:3}, we get
\begin{equation}
\kappa_{H}(a,d)\geq |A|+|D|+|\bar{B}_{d}|+|B_{a}|+|B_{d}\cap\bar{B}_{a}|+|C_{d}|+|\bar{C}_{a}|+|C_{a}\cap\bar{C}_{d}|.
\label{eq:leastconn}
\end{equation}

We now prove the second part of the lemma. If $a$ and $d$ are not adjacent in $G_{4p}$, they are not adjacent in $G$ either, so there is no 4-clique in $G$ containing them. Otherwise, $a$ and $d$ are adjacent, we claim that there is a 4-clique containing $a$ and $d$ in $G_{4p}$ (which is equivalent to the existence of a 4-clique containing $a$ and $d$ in $G$ by \Cref{def:4partite}) if and only if $\kappa_{\hat{H}_{ad}}(a,d)\geq 1$. If the 4-clique exists, let $b\in B$ and $c\in C$ be the other vertices in the 4-clique. Then there is a path $a\to b\to c\to d$ in $\hat{H}_{ad}$ from the construction, which implies $\kappa_{\hat{H}_{ad}}(a,d)\geq 1$. On the other hand, if $\kappa_{\hat{H}_{ad}}(a,d)\geq 1$, there is a path $a\to b\to c\to d$, and $(a,b,c,d)$ forms a 4-clique in $G_{4p}$. Combining this claim with \Cref{eq:1,eq:2,eq:3} gives that there is a 4-clique containing $a$ and $d$ in $G$ if and only if
\begin{equation}
\kappa_{H}(a,d)\geq |A|+|D|+|\bar{B}_{d}|+|B_{a}|+|B_{d}\cap\bar{B}_{a}|+|C_{d}|+|\bar{C}_{a}|+|C_{a}\cap\bar{C}_{d}|+1.
\label{eq:threshold}
\end{equation}
Finally, by the construction of $G_{4p}$, we have $|A|=|B|=|C|=|D|=n$, $|B_{a}|=|C_{a}|$, $|B_{d}|=|C_{d}|$, $|B_{d}\cap \bar{B}_{a}|=|N_{G}(d)\cap\bar{N}_{G}(a)|$ and $|C_{a}\cap\bar{C}_{d}|=|N_{G}(a)\cap\bar{N}_{G}(d)|$. Combining them with Inequalities (\ref{eq:leastconn}) and (\ref{eq:threshold}) completes the proof.

\end{proof}

\begin{proof}[Proof of \Cref{thm:APVClowerbound}]
Assume for contradiction that there exists a combinatorial algorithm $\mathcal{A}$ for the APVC problem with running time $O(n^{4-\epsilon})$ for some constant $\epsilon>0$. Let $G$ be an arbitrary 4-clique instance. We first construct the 4-partite graph $G_{4p}$ and the graph $H$ following the construction in this section. Note that $V(H)=V(Q)\cup A\cup B\cup B'\cup C\cup C'\cup D$, so $|V(H)|=O(n)$ by the construction and \Cref{lemma:CleanupGadget}. Also, $H$ can be constructed in $O(n^{2})$ time directly. By \Cref{lemma:reduction}, we can solve the 4-clique problem by first running $\mathcal{A}$ on graph $H$ and then checking for each adjacent $a\in A,d\in D$ whether $\kappa_{H}(a,d)$ reaches the threshold value $4n+|N_{G}(a)\cap \bar{N}_{G}(d)| + |\bar{N}_{G}(a)\cap N_{G}(d)|+1$. This takes $O(n^{4-\epsilon})+O(n^{3})=O(n^{4-\epsilon})$ time because computing the threshold value for each $a\in A,d\in D$ takes totally $O(n^{3})$ extra time. This contradicts \Cref{conj:4clique}.

\end{proof}

\begin{remark}
The proof of \Cref{thm:APVClowerbound} basically shows that if the APVC problem can be solved in time $T_{\APVC}(n)$, then the 4-clique problem can be solved in time $T_{\fclique}(n)=O(T_{\APVC}(n)+n^{3})$. We note that this relation can be improved to $T_{\fclique}(n)=O(T_{\APVC}(n)+n^{\omega})$ if we use fast matrix multiplication to speed up the reduction. For general algorithms, another version of the 4-clique conjecture (see e.g. \cite{DV22}) states that solving the 4-clique problem requires $n^{\omega(1,2,1)-o(1)}$ time. Therefore, assuming this conjecture, solving the APVC problem also requires $n^{\omega(1,2,1)-o(1)}$ time for general algorithms.
\label{remark:generalAPVC}
\end{remark}

\subsection{Further Results}

In this section, we will show several corollaries from the lower bound of the APVC problem. 

The first corollary is a conditional lower bound of the SSVC problem.

\begin{corollary}

\label{coro:SSVC}

Assuming \Cref{conj:4clique}, for $n$-vertex undirected unweighted graphs, there is no combinatorial algorithm that solves the SSVC problem in $O(n^{3-\epsilon})$ time for any constant $\epsilon > 0$.
\end{corollary}
\begin{proof}
Assume for the contradiction that the SSVC problem can be solved in $O(n^{3-\epsilon})$ time. Then for an APVC instance $G$, we may treat every vertex in $G$ as a source to get the correct output. It takes $O(n)$ SSVC calls and therefore the complexity is $O(n \cdot n^{3-\epsilon})=O(n^{4-\epsilon})$, which contradicts \Cref{thm:APVClowerbound} assuming \Cref{conj:4clique}.
\end{proof}

The second corollary is a conditional lower bound of the APVC problem for graphs with general density.

\begin{corollary}

\label{coro:Density}

Given any constant $\delta\in[0,1]$, assuming \Cref{conj:4clique}, there is no combinatorial algorithm that solves the APVC problem for $n$-vertex $m$-edge unweighted graphs, where $m=\Theta(n^{1+\delta})$ with running time $O(m^{2-\epsilon})$ for any constant $\epsilon > 0$.

\end{corollary}

\begin{proof}
Assume that for some constants $\delta$ and $\epsilon$, such $O(m^{2-\epsilon})$-time algorithm $\mathcal{A}$ exists. Let $H$ be an $\hat{n}$-vertex $\hat{m}$-edge APVC hard instance with $\hat{m}=\Theta(\hat{n}^{2})$, constructed as above for some 4-clique instance. Let $G$ be the union of $H$ and $\Theta(\hat{m}^{1/(1+\delta)})$ isolated vertices. Observe that $G$ now has $n=\hat{n}+\Theta(\hat{m}^{1/(1+\delta)})$ vertices and $m=\hat{m}$ edges, i.e. $m=\Theta(n^{1+\delta})$. By applying algorithm $\mathcal{A}$ on $G$, the all-pairs vertex connectivity of $H$ can be computed in $O(m^{2-\epsilon})$, i.e. $O(\hat{n}^{4-2\epsilon})$ time, contradicting \Cref{conj:4clique} by the argument in the proof of \Cref{thm:APVClowerbound}.
\end{proof}

The last corollary is a conditional lower bound of the SSVC problem for graphs with general density. The proof is omitted since it is analogous to \Cref{coro:Density}.

\begin{corollary}

\label{coro:SSVCDensity}
Given any constant $\delta\in[0,1]$, assuming \Cref{conj:4clique}, there is no combinatorial algorithm that solves the SSVC problem for $n$-vertex $m$-edge unweighted graphs, where $m=\Theta(n^{1+\delta})$, with running time $O(m^{3/2-\epsilon})$ for any constant $\epsilon > 0$.
\end{corollary}

\section{The Lower Bound for Steiner Vertex Connectivity Problem}
\label{sect:SteinerLowerBound}

In this section, we will prove \Cref{thm:Steinerlowerbound}, a conditional lower bound of the Steiner vertex connectivity problem in undirected unweighted graphs, conditioning on the edge-universal 4-clique problem.

\begin{theorem}
For $n$-vertex undirected unweighted graphs, assuming \Cref{conj:AllEdge4Clique}, there is no combinatorial algorithm that solves the Steiner vertex connectivity problem in $O(n^{4-\epsilon})$ time for any constant $\epsilon > 0$. 
\label{thm:Steinerlowerbound}
\end{theorem}

Given an $n$-vertex edge-universal 4-clique instance $G$ with a set $E_{\dem}$ of \emph{demand} edges, let $H$ be the hard APVC instance we construct in \Cref{sect:APVClowerbound}. We will strengthen $H$ to another graph $J$ such that the edge-universal 4-clique problem in $G$ can be reduced to a Steiner vertex connectivity problem in $J$ of the terminal set $A\cup D$. As mentioned in \Cref{sect:overview}, the main idea is to add extra ``flow'' paths from $A$ to $D$ to make the additive deviations between $\kappa_{J}(a,d)$ and $\kappa_{\hat{H}_{ad}}(a,d)$ uniform for all pairs of $a\in A$ and $d\in D$. Furthermore, to exclude the interference of vertex connectivity of pairs $(a_{1}, a_{2})$ s.t. $a_{1}, a_{2}\in A$ or $(d_{1},d_{2})$ s.t. $d_{1},d_{2}\in D$, we artificially add large additive deviations for these pairs.

The construction of $J$ is as follows and see \Cref{fig:ST} for an illustration. The vertex set $V(J)=V(H)\cup Z\cup W\cup A'\cup D'$, where $Z, W, A', D'$ are disjoint groups of additional vertices to create extra ``flow'' paths. Concretely, $Z$ and $W$ will be copies of original vertex set $V(G)$, and $A'$ and $D'$ are additional sets of vertices of size $|A'|=|D'|=10n$. Let 
\begin{align*}
E_{Z} =& \{(a, z)|a\in A, z\in Z,(a,z) \in E(G)\} \cup \\ &\{(d, z)|d\in D,z\in Z,(d,z) \in E(G)\}
\end{align*}
and
\begin{align*}
E_{W} =& \{(a, w)|a\in A, w\in W,(a,w) \not \in E(G)\} \cup \\ &\{(d, w)|d\in D,w\in W,(d,w) \not \in E(G)\}
\end{align*}
be the extra edges to ``equalize the deviation'' of all pairs $(a,d)$ between $A$ and $D$. Let
\[
    E_{A'}=\{(a,a')\mid a\in A,a'\in A'\}
\]
and
\[
    E_{D'}=\{(d,d')\mid d\in D,d'\in D'\}
\]
be extra edges that bring large deviations to pairs inside $A$ or $D$. Finally, we construct a set of extra edges
\[
    E_{AD}=\{(a,d)\mid a\in A,d\in D,(a,d)\notin E_{\dem}\}
\]
to prevent non-demand pairs $(V(G)\times V(G))\setminus E_{\dem}$ from affecting the Steiner connectivity value. The whole edge set $E(J)=V(H)\cup E_{Z}\cup E_{W}\cup E_{A'}\cup E_{D'}\cup E_{AD}$.

\begin{figure}[ht]
    \centering
    \includegraphics[width=1\textwidth]{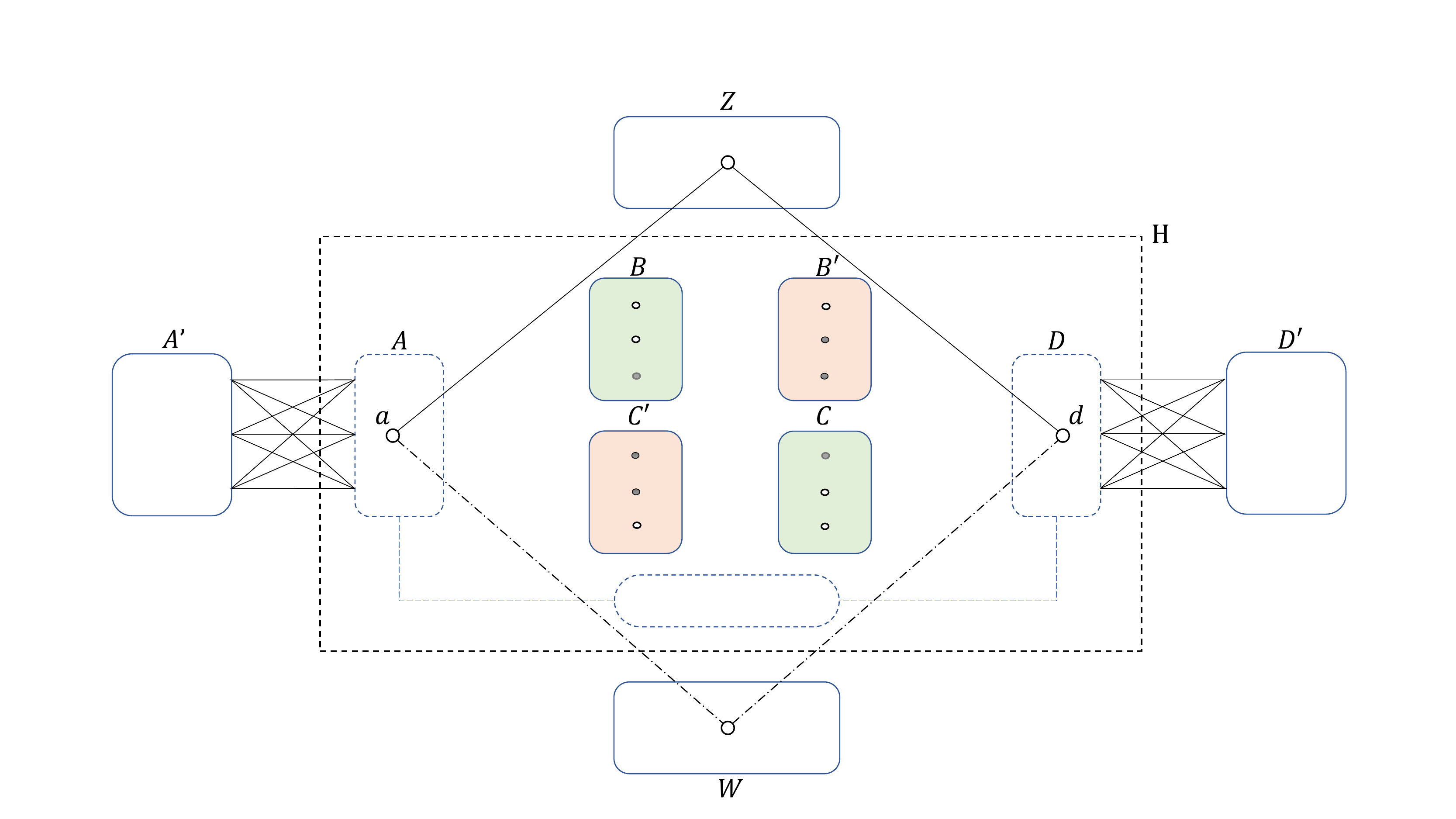}
    \caption{
    The graph $J$. 
    There are bipartite cliques between $A$ and $A'$, $D$ and $D'$, and the edges between $A,D$ and $Z,W$ are linked based on $E_Z,E_W$.
    }
    \label{fig:ST}
\end{figure}

The correctness of the reduction will be established by the following lemmas.

\begin{lemma}
\label{lemma:steiner1}
In graph $J$, for each $a_{1},a_{2}\in A$ with $a_{1}\neq a_{2}$ and $d_{1},d_{2}\in D$ with $d_{1}\neq d_{2}$, we have $\kappa_{J}(a_{1},a_{2})\geq 10n$ and $\kappa_{J}(d_{1},d_{2})\geq 10n$.
\end{lemma}

\begin{proof}

For each pair $a_{1},a_{2}\in A$ with $a_{1}\neq a_{2}$, we can construct at least $10n$ internally vertex-disjoint paths from $a_{1}$ to $a_{2}$ via vertices in $A'$ and edges in $E_{A'}$, so $\kappa_{J}(a_{1},a_{2})\geq 10n$. Similarly, $\kappa_{J}(d_{1},d_{2})\geq 10n$.

\end{proof}

\begin{lemma}
\label{lemma:steiner2}
In graph $J$, for each $a \in A$, $d \in D$ with $(a,d) \not \in E_{\dem}$, $\kappa_{J}(a, d) \geq 5n +1$.
\end{lemma}

\begin{proof}

From our construction, the number of nodes in $Z$ adjacent to both $a$ and $d$ is 
\[
|\{z \in Z| (a,z) \in E_Z, (d,z)\in E_Z\}| = |N_G(a) \cap N_G(d)|,
\]
and the number of nodes in $W$ adjacent to both $a$ and $d$ is
\[
|\{w \in W| (a,w) \in E_Z, (d,w)\in E_Z\}| = |\Bar{N}_G(a) \cap \Bar{N}_G(d)|.\]
Furthermore, there is an individual edge in $E_{AD}$ directly connecting $a$ and $d$ for each such $(a,d)\notin E_{\dem}$.

Therefore, we have $|N_G(a) \cap N_G(d)| + |\Bar{N}_G(a) \cap \Bar{N}_G(d)|+1$ extra vertex-disjoint paths from $a$ to $d$. Furthermore, Inequality (\ref{eq:ConnLB}) from \Cref{lemma:reduction} shows that $\kappa_{H}(a,d)\geq 4n + |N_G(a) \cap \Bar{N}_G(d)| + |\bar{N}_G(a) \cap N_G(d)|$, so we can construct a vertex-disjoint path set with size at least
\begin{align*}
4n &+ |N_G(a) \cap \Bar{N}_G(d)| + |\bar{N}_G(a) \cap N_G(d)| \\ &+ |N_G(a) \cap N_G(d)| + |\Bar{N}_G(a) \cap \Bar{N}_G(d)|+1 = 5n+1,
\end{align*}
which implies $\kappa_{J}(a,d)\geq 5n+1$.

\end{proof}

\begin{lemma}
\label{lemma:steiner3}
In graph $J$, for each $a \in A$, $d \in D$ with $(a,d) \in E_{\dem}$, there is a 4-clique containing $a$ and $d$ in $G$ if and only if $\kappa_{J}(a,d) \geq 5n+1$.

\end{lemma}

\begin{proof}

First, assume there is a 4-clique containing $a$ and $d$ in $G$. From \Cref{lemma:reduction}, we can construct $4n+|N_{G}(a)\cap \bar{N}_{G}(d)| + |\bar{N}_{G}(a)\cap N_{G}(d)|+1$ internally vertex-disjoint paths connecting $a$ and $d$ in $H$. An argument similar to the proof of \Cref{lemma:steiner2} shows that we have extra $|N_G(a) \cap N_G(d)| + |\Bar{N}_G(a) \cap \Bar{N}_G(d)|$ internally vertex-disjoint paths via vertices in $Z,W$ and edges in $E_{Z},E_{W}$. Therefore, in graph $J$, we can find a set of internally vertex-disjoint paths from $a$ to $d$ with size
\begin{align*}
4n &+ |N_G(d) \cap \Bar{N_G}(a)| + |N_G(a) \cap \Bar{N_G}(d)| + 1 \\ &+ |N_G(a) \cap N_G(d)| + |\Bar{N_G}(a) \cap \Bar{N_G}(d)| = 5n+1.
\end{align*}

Now assume there is no 4-clique containing $a$ and $d$ in $G$. From \Cref{lemma:CleanupGadget}, we define $J_{ad}=J \setminus ((A \cup D) \setminus \{a,d\})$ by removing vertices in $A$ and $D$ other than $a$ and $d$, and then we have 
\begin{equation}
\kappa_{J}(a,d)=\kappa_{J_{ad}}(a,d)+|A|+|D|.
\label{eq:8}
\end{equation}
Note that although the source-sink isolating gadget $Q(A,D)$ is constructed by applying \Cref{lemma:CleanupGadget} on graph $H$ and vertex sets $A,D$, the construction of $Q(A,D)$ is independent from $H$, so we can also apply \Cref{lemma:CleanupGadget} on graph $J$ and vertex sets $A,D$.

Next, we will show
\begin{equation}
\kappa_{J_{ad}}(a,d) \leq \kappa_{H_{ad}}(a,d) + |N_G(a) \cap N_G(d)| + |\Bar{N}_G(a) \cap \Bar{N}_G(d)|,
\label{eq:9}
\end{equation}

where $H_{ad}=H \setminus ((A \cup D)\setminus \{a,d\})$ as defined in the proof of \Cref{lemma:reduction}. The reason is that, compared to $H_{ad}$, the extra vertices in $J_{ad}$ are in $A', D', Z, W$, which are only directly connected to $a$ or $d$. Let $C_{H}$ be the minimum vertex cut of size $|C_{H_{ad}}|=\kappa_{H_{ad}}(a,d)$ whose removal disconnects $a$ and $d$ in $H_{ad}$. Let 
\begin{align*}
C_{J_{ad}} = & C_{H_{ad}}\cup \{z \in Z\mid (a,z) \in E_Z, (d,z)\in E_Z\} \cup\\ & \{w \in W\mid  (a,w) \in E_Z, (d,w)\in E_Z\}.
\end{align*}
We can easily observe that removing $C_{J_{ad}}$ will disconnect $a$ and $d$ in $J_{ad}$, which immediately implies Inequality (\ref{eq:9}).

Finally, combining Equation (\ref{eq:8}) and Inequality (\ref{eq:9}), we have
\begin{align*}
    \kappa_{J}(a,d)\leq& \kappa_{H_{ad}}(a,d)+|A|+|D|+ |N_G(a) \cap N_G(d)| \\ &+ |\Bar{N}_G(a) \cap \Bar{N}_G(d)|.
\end{align*}
From \Cref{lemma:reduction}, we know that when there is no 4-clique containing $a$ and $d$ in $G$, $\kappa_{H}(a,d)=4n+|N_{G}(a)\cap \bar{N}_{G}(d)| + |\bar{N}_{G}(a)\cap N_{G}(d)|$. Combining $\kappa_{H}(a,d)=\kappa_{H_{ad}}(a,d)+|A|+|D|$ (by \Cref{lemma:CleanupGadget} again),
\begin{align*}
\kappa_{J}(a,d) \leq& \kappa_{H_{ad}}(a,d)+|A|+|D| \\ &+ |N_G(a) \cap N_G(d)| + |\Bar{N}_G(a) \cap \Bar{N}_G(d)|\\
=&\kappa_{H}(a,d)+|N_G(a) \cap N_G(d)| + |\Bar{N}_G(a) \cap \Bar{N}_G(d)|\\
=&4n+|N_{G}(a)\cap \bar{N}_{G}(d)| + |\bar{N}_{G}(a)\cap N_{G}(d)| \\ &+ |N_G(a) \cap N_G(d)| + |\Bar{N}_G(a) \cap \Bar{N}_G(d)|\\
=&5n<5n+1.
\end{align*}
\end{proof}

We are now ready to prove \Cref{thm:Steinerlowerbound}.

\begin{proof}[Proof of \Cref{thm:Steinerlowerbound}]

Given an edge-universal 4-clique instance $G$ with a set $E_{\dem}$ of demand edges. We first construct the Steiner vertex connectivity instance $J$ with terminal set $A\cup D$ as shown above. Then $G$ with $E_{\dem}$ is a yes instance of the edge-universal 4-clique problem if and only if the Steiner vertex connectivity of $A\cup D$ in $J$ is at least $5n+1$, by \Cref{lemma:steiner1,lemma:steiner2,lemma:steiner3}.

For time analysis, note that $V(J)=O(n)$ by construction. Therefore, an $O(n^{4-\epsilon})$-time combinatorial algorithm for the Steiner vertex connectivity problem will imply an $O(n^{4-\epsilon})$ time combinatorial algorithm for the edge-universal 4-clique problem, which contradicts \Cref{conj:AllEdge4Clique}.

\end{proof}

\section{The Upper Bounds}
\label{sect:UpperBound}

In this section, we will show the upper bounds of the APVC problem and SSVC problem in undirected unweighted sparse graphs (see \Cref{thm:APVCupper} and \Cref{thm:SSVCupper} respectively). We only prove \Cref{thm:APVCupper} in detail and briefly mention the proof of \Cref{thm:SSVCupper} since they are analogous.

\begin{theorem}

Given an undirected unweighted graph $G$ with $n$ vertices and $m$ edges, there is a randomized algorithm that solves the APVC problem in $\hat O(m^{\frac{11}{5}})$ time with high probability. Furthermore, assuming \Cref{conj:GHtree}, the running time of this algorithm becomes $\hat O(m^2)$.

\label{thm:APVCupper}

\end{theorem}

\begin{theorem}

Given an undirected unweighted graph $G$ with $n$ vertices, $m$ edges and a source vertex, there is a randomized algorithm to solve the SSVC problem in $\hat O(m^{\frac{5}{3}})$ time with high probability. Furthermore, assuming \Cref{conj:GHtree}, the running time becomes $\hat O(m^{1.5})$.

\label{thm:SSVCupper}

\end{theorem}

A key subroutine to obtain the algorithm in \Cref{thm:APVCupper} is the subroutine shown in \Cref{lemma:SmallVC}, which can capture all vertex connectivity bounded by $k$. The proof of \Cref{lemma:SmallVC} has been shown in \cite{IN12,PSY22} and we defer it to \Cref{sect:OmittedProof} for completeness.

\begin{lemma}

Given an $n$-vertex $m$-edge undirected unweighted graph $G$ and a threshold $k$, there is a randomized algorithm that computes the value of $\min\{\kappa_{G}(u,v),k\}$ for all vertex pairs $(u,v)$ in $\hat{O}(mk^3+n^{2})$ time with high probability. Furthermore, assuming \Cref{conj:GHtree}, the running time of this algorithm becomes $\hat{O}(mk^2+n^{2})$.

\label{lemma:SmallVC}

\end{lemma}

We now complete the proof of \Cref{thm:APVCupper} using \Cref{lemma:SmallVC}.

\begin{proof}[Proof of \Cref{thm:APVCupper}]
As discussed in \Cref{sect:overview}, our APVC algorithm will handle vertex connectivity between high-degree vertices and vertex connectivity involving low-degree vertices separately. Let $k$ be a degree threshold which will be chosen later. Given an input graph $G$, let $V_{h}=\{v\in V(G)\mid \deg_{G}(v)>k\}$ and $V_{l}=\{v\in V(G)\mid \deg_{G}(v)\leq k\}$. 

Note that for each vertex pair $(u,v)$ such that $u\in V_{l}$ or $v\in V_{l}$, the vertex connectivity $\kappa_{G}(u,v)$ is at most $k$. By \Cref{lemma:SmallVC}, the vertex connectivity of all such pairs can be computed exactly in time $\hat{O}(mk^{3}+n^2)$. Now consider all remaining pairs $(u,v)$ with $u\in V_{h}$ and $v\in V_{h}$. Since $|V_{h}|\leq O(m/k)$ by definition, there are $O(m^{2}/k^{2})$ remaining pairs. The vertex connectivity of each pair can be computed in $\hat{O}(m)$ time using one max flow call \cite{CKL+22}, so totally $\hat{O}(m^{3}/k^{2})$ time suffices for remaining pairs.

The total running time of the above algorithm is then $\hat{O}(mk^{3}+n^{2}+m^{3}/k^{2})$, by choosing $k=\Theta(m^{2/5})$, the running time will be $\hat{O}(m^{11/5})$. Assuming \Cref{conj:GHtree}, the running time will be improved to $\hat{O}(mk^{2}+n^{2}+m^{3}/k^{2})$, which is $\hat{O}(m^{2})$ by choosing $k=\Theta(m^{1/2})$.
\end{proof}

The proof of \Cref{thm:SSVCupper} is analogous and we sketch it below.

\begin{proof}[Proof of \Cref{thm:SSVCupper}]

Let $s\in V(G)$ be the source. Similarly, we partition $V(G)$ into two set $V_{h}=\{v\in V(G)\mid \deg_{G}(v)>k\}$ and $V_{l}=\{v\in V(G)\mid \deg_{G}(v)\leq k\}$ where $k$ is the degree threshold. The vertex connectivity $\kappa_{G}(s,v)$ for all pairs $(s,v)$ such that $v\in V_{l}$ can be computed in $\hat{O}(mk^{2}+n)$ time, using a subroutine analogous to \Cref{lemma:SmallVC}. The vertex connectivity of remaining pairs can be computed in $\hat{O}(m^{2}/k)$ by trivially calling max flows.

The total running time is $\hat{O}(mk^{2}+n+m^{2}/k)$, which will be $\hat{O}(m^{5/3})$ by choosing $k=\Theta(m^{1/3})$. Assuming \Cref{conj:GHtree}, the running time is $\hat{O}(mk+n+m^{2}/k)$, which will be $\hat{O}(m^{3/2})$ by choosing $k=\Theta(m^{1/2})$.

\end{proof}

\subsection{Proof of \Cref{lemma:SmallVC}}
\label{sect:OmittedProof}
The algorithm and analysis follow the ideas in \cite{IN12}.

The algorithm is as follows. First, we create $t=O(k^{2}\log n)$ sample sets $U_{1},...,U_{t}$, each of which is generated by sampling each vertex in $V(G)$ independently with probability $1/k$. Moreover, for each set $U_{i}$, we construct a $k$-Gomory-Hu tree for element connectivity using \Cref{thm:GHtree}. From \Cref{thm:GHtree}, for each set $U_{i}$ and each pair $u,v\in U_{i}$, we can query $a_{i}(u,v)=\min\{\kappa'_{G,U_{i}}(u,v),k\}$ in nearly constant time. Then for each $u,v\in V(G)$, we let the final output be $a(u,v)=\min\{a_{i}(u,v)\mid u,v\in U_{i}\}$. 

We then analyze the running time. The time to construct all $k$-Gomory-Hu trees is $\hat{O}(t\cdot mk)=\hat{O}(mk^{3})$. To compute all $a_{i}(u,v)$ and final output $a(u,v)$, each set $U_{i}$ will have size $\tilde{O}(n/k)$ w.h.p. by Chernoff bound, so the time is $\tilde{O}(t\cdot (n/t)^{2})=\tilde{O}(n^{2})$ (by aborting the algorithm when some $U_{i}$ has size not bounded by $\tilde{O}(n/k)$). Therefore, the total running time is $\hat{O}(mk^{3}+n^{2})$.

The correctness is shown as follows. Consider a fixed $u,v \in V(G)$. First, $a(u,v)$ is well-defined with high probability, because for one sample set $U_{i}$, $u$ and $v$ are inside $U_{i}$ and $a_{i}(u,v)$ is well-defined with probability $1/k^{2}$ and there are $O(k^{2}\log n)$ sample sets. Given that $a(u,v)$ is well-defined, we know $a(u,v)\geq \min\{\kappa_{G}(u,v),k\}$ since $\kappa'_{G,U_{i}}(u,v)\geq \kappa_{G}(u,v)$ for all $U_{i}$ by the definition of element connectivity. By the same reason, if $\kappa_{G}(u,v)\geq k$, we must have $a(u,v)=k=\min\{\kappa_{G}(u,v),k\}$ which is the correct answer. From now we suppose $\kappa_{G}(u,v)<k$ and we are going to show there exists $U_{i}$ such that $u,v\in U_{i}$ and $a_{i}(u,v)=\kappa_{G}(u,v)$ with high probability. Note that $\kappa'_{G,U_{i}}(u,v)=\kappa$ if $U_{i}$ is disjoint with some minimum $u$-$v$ vertex cut $C_{u,v}$. From our sampling strategy, $U_{i}$ contains $u,v$ and is disjoint with $C_{u,v}$ with probability $\frac{1}{k^{2}}(1-\frac{1}{k})^{k}=\Omega(1/k^{2})$. Because there are $O(k^{2}\log n)$ sample set, $U_{i}$ exists with high probability.

If we assuming \Cref{conj:GHtree}, we simply subsitute $k$-Gomory-Hu tree for element connectivity with element connectivity Gomory-Hu tree and follow the same algorithm and analysis. We will obtain an algorithm with running time $\hat{O}(mk^{2}+n^{2})$.


\newcommand{\etalchar}[1]{$^{#1}$}


\begin{thebibliography}{vdBLN{\etalchar{+}}20}

\bibitem[ABHS22]{abboud2022scheduling}
Amir Abboud, Karl Bringmann, Danny Hermelin, and Dvir Shabtay.
\newblock Scheduling lower bounds via {AND} subset sum.
\newblock {\em J. Comput. Syst. Sci.}, 127:29--40, 2022.

\bibitem[AGI{\etalchar{+}}19]{abboud2018faster}
Amir Abboud, Loukas Georgiadis, Giuseppe~F. Italiano, Robert Krauthgamer, Nikos
  Parotsidis, Ohad Trabelsi, Przemyslaw Uznanski, and Daniel Wolleb{-}Graf.
\newblock Faster algorithms for all-pairs bounded min-cuts.
\newblock In Christel Baier, Ioannis Chatzigiannakis, Paola Flocchini, and
  Stefano Leonardi, editors, {\em 46th International Colloquium on Automata,
  Languages, and Programming, {ICALP} 2019, July 9-12, 2019, Patras, Greece},
  volume 132 of {\em LIPIcs}, pages 7:1--7:15. Schloss Dagstuhl -
  Leibniz-Zentrum f{\"{u}}r Informatik, 2019.

\bibitem[AKL{\etalchar{+}}22]{AKL+22}
Amir Abboud, Robert Krauthgamer, Jason Li, Debmalya Panigrahi, Thatchaphol
  Saranurak, and Ohad Trabelsi.
\newblock Breaking the cubic barrier for all-pairs max-flow: Gomory-hu tree in
  nearly quadratic time.
\newblock In {\em 63rd {IEEE} Annual Symposium on Foundations of Computer
  Science, {FOCS} 2022, Denver, CO, USA, October 31 - November 3, 2022}, pages
  884--895. {IEEE}, 2022.

\bibitem[AKT20]{AKT20}
Amir Abboud, Robert Krauthgamer, and Ohad Trabelsi.
\newblock New algorithms and lower bounds for all-pairs max-flow in undirected
  graphs.
\newblock In Shuchi Chawla, editor, {\em Proceedings of the 2020 {ACM-SIAM}
  Symposium on Discrete Algorithms, {SODA} 2020, Salt Lake City, UT, USA,
  January 5-8, 2020}, pages 48--61. {SIAM}, 2020.

\bibitem[AKT21a]{AKT22}
Amir Abboud, Robert Krauthgamer, and Ohad Trabelsi.
\newblock {APMF} {\textless} apsp? gomory-hu tree for unweighted graphs in
  almost-quadratic time.
\newblock In {\em 62nd {IEEE} Annual Symposium on Foundations of Computer
  Science, {FOCS} 2021, Denver, CO, USA, February 7-10, 2022}, pages
  1135--1146. {IEEE}, 2021.

\bibitem[AKT21b]{AKT21}
Amir Abboud, Robert Krauthgamer, and Ohad Trabelsi.
\newblock Subcubic algorithms for gomory-hu tree in unweighted graphs.
\newblock In Samir Khuller and Virginia~Vassilevska Williams, editors, {\em
  {STOC} '21: 53rd Annual {ACM} {SIGACT} Symposium on Theory of Computing,
  Virtual Event, Italy, June 21-25, 2021}, pages 1725--1737. {ACM}, 2021.

\bibitem[AWW16]{abboud2016approximation}
Amir Abboud, Virginia~Vassilevska Williams, and Joshua~R. Wang.
\newblock Approximation and fixed parameter subquadratic algorithms for radius
  and diameter in sparse graphs.
\newblock In Robert Krauthgamer, editor, {\em Proceedings of the Twenty-Seventh
  Annual {ACM-SIAM} Symposium on Discrete Algorithms, {SODA} 2016, Arlington,
  VA, USA, January 10-12, 2016}, pages 377--391. {SIAM}, 2016.

\bibitem[Ben95]{Ben95}
Andr{\'{a}}s~A. Bencz{\'{u}}r.
\newblock Counterexamples for directed and node capacitated cut-trees.
\newblock {\em {SIAM} J. Comput.}, 24(3):505--510, 1995.

\bibitem[BGL17]{BGL17}
Karl Bringmann, Allan Gr{\o}nlund, and Kasper~Green Larsen.
\newblock A dichotomy for regular expression membership testing.
\newblock In Chris Umans, editor, {\em 58th {IEEE} Annual Symposium on
  Foundations of Computer Science, {FOCS} 2017, Berkeley, CA, USA, October
  15-17, 2017}, pages 307--318. {IEEE} Computer Society, 2017.

\bibitem[CKL{\etalchar{+}}22]{CKL+22}
Li~Chen, Rasmus Kyng, Yang~P. Liu, Richard Peng, Maximilian~Probst Gutenberg,
  and Sushant Sachdeva.
\newblock Maximum flow and minimum-cost flow in almost-linear time.
\newblock In {\em 63rd {IEEE} Annual Symposium on Foundations of Computer
  Science, {FOCS} 2022, Denver, CO, USA, October 31 - November 3, 2022}, pages
  612--623. {IEEE}, 2022.

\bibitem[CKM{\etalchar{+}}11]{CKMST11}
Paul~F. Christiano, Jonathan~A. Kelner, Aleksander Madry, Daniel~A. Spielman,
  and Shang{-}Hua Teng.
\newblock Electrical flows, laplacian systems, and faster approximation of
  maximum flow in undirected graphs.
\newblock In Lance Fortnow and Salil~P. Vadhan, editors, {\em Proceedings of
  the 43rd {ACM} Symposium on Theory of Computing, {STOC} 2011, San Jose, CA,
  USA, 6-8 June 2011}, pages 273--282. {ACM}, 2011.

\bibitem[CRX15]{CRX15}
Chandra Chekuri, Thapanapong Rukkanchanunt, and Chao Xu.
\newblock On element-connectivity preserving graph simplification.
\newblock In Nikhil Bansal and Irene Finocchi, editors, {\em Algorithms - {ESA}
  2015 - 23rd Annual European Symposium, Patras, Greece, September 14-16, 2015,
  Proceedings}, volume 9294 of {\em Lecture Notes in Computer Science}, pages
  313--324. Springer, 2015.

\bibitem[DW22]{DV22}
Mina Dalirrooyfard and Virginia~Vassilevska Williams.
\newblock Induced cycles and paths are harder than you think.
\newblock In {\em 63rd {IEEE} Annual Symposium on Foundations of Computer
  Science, {FOCS} 2022, Denver, CO, USA, October 31 - November 3, 2022}, pages
  531--542. {IEEE}, 2022.

\bibitem[FF56]{FF56}
L.~R. Ford and D.~R. Fulkerson.
\newblock Maximal flow through a network.
\newblock {\em Canadian Journal of Mathematics}, 8:399–404, 1956.

\bibitem[GH61]{GH61}
R.~E. Gomory and T.~C. Hu.
\newblock Multi-terminal network flows.
\newblock {\em Journal of the Society for Industrial and Applied Mathematics},
  9(4):551--570, 1961.

\bibitem[GIKW19]{gao2018completeness}
Jiawei Gao, Russell Impagliazzo, Antonina Kolokolova, and Ryan Williams.
\newblock Completeness for first-order properties on sparse structures with
  algorithmic applications.
\newblock {\em {ACM} Trans. Algorithms}, 15(2):23:1--23:35, 2019.

\bibitem[GR98]{GR98}
Andrew~V. Goldberg and Satish Rao.
\newblock Beyond the flow decomposition barrier.
\newblock {\em J. {ACM}}, 45(5):783--797, 1998.

\bibitem[HRG00]{HRG00}
Monika~Rauch Henzinger, Satish Rao, and Harold~N. Gabow.
\newblock Computing vertex connectivity: New bounds from old techniques.
\newblock {\em J. Algorithms}, 34(2):222--250, 2000.

\bibitem[IN12]{IN12}
Rani Izsak and Zeev Nutov.
\newblock A note on labeling schemes for graph connectivity.
\newblock {\em Inf. Process. Lett.}, 112(1-2):39--43, 2012.

\bibitem[Kar00]{Kar00}
David~R. Karger.
\newblock Minimum cuts in near-linear time.
\newblock {\em J. {ACM}}, 47(1):46--76, 2000.

\bibitem[Kle69]{Kle69}
D.~Kleitman.
\newblock Methods for investigating connectivity of large graphs.
\newblock {\em IEEE Transactions on Circuit Theory}, 16(2):232--233, 1969.

\bibitem[KLOS14]{KLOS14}
Jonathan~A. Kelner, Yin~Tat Lee, Lorenzo Orecchia, and Aaron Sidford.
\newblock An almost-linear-time algorithm for approximate max flow in
  undirected graphs, and its multicommodity generalizations.
\newblock In Chandra Chekuri, editor, {\em Proceedings of the Twenty-Fifth
  Annual {ACM-SIAM} Symposium on Discrete Algorithms, {SODA} 2014, Portland,
  Oregon, USA, January 5-7, 2014}, pages 217--226. {SIAM}, 2014.

\bibitem[LLW88]{linial1988rubber}
Nathan Linial, L{\'{a}}szl{\'{o}} Lov{\'{a}}sz, and Avi Wigderson.
\newblock Rubber bands, convex embeddings and graph connectivity.
\newblock {\em Comb.}, 8(1):91--102, 1988.

\bibitem[LNP{\etalchar{+}}21]{LNP+21}
Jason Li, Danupon Nanongkai, Debmalya Panigrahi, Thatchaphol Saranurak, and
  Sorrachai Yingchareonthawornchai.
\newblock Vertex connectivity in poly-logarithmic max-flows.
\newblock In Samir Khuller and Virginia~Vassilevska Williams, editors, {\em
  {STOC} '21: 53rd Annual {ACM} {SIGACT} Symposium on Theory of Computing,
  Virtual Event, Italy, June 21-25, 2021}, pages 317--329. {ACM}, 2021.

\bibitem[LP20]{LP20}
Jason Li and Debmalya Panigrahi.
\newblock Deterministic min-cut in poly-logarithmic max-flows.
\newblock In Sandy Irani, editor, {\em 61st {IEEE} Annual Symposium on
  Foundations of Computer Science, {FOCS} 2020, Durham, NC, USA, November
  16-19, 2020}, pages 85--92. {IEEE}, 2020.

\bibitem[LP21]{LP21}
Jason Li and Debmalya Panigrahi.
\newblock Approximate gomory-hu tree is faster than \emph{n} - 1 max-flows.
\newblock In Samir Khuller and Virginia~Vassilevska Williams, editors, {\em
  {STOC} '21: 53rd Annual {ACM} {SIGACT} Symposium on Theory of Computing,
  Virtual Event, Italy, June 21-25, 2021}, pages 1738--1748. {ACM}, 2021.

\bibitem[LPS21]{LPS22}
Jason Li, Debmalya Panigrahi, and Thatchaphol Saranurak.
\newblock A nearly optimal all-pairs min-cuts algorithm in simple graphs.
\newblock In {\em 62nd {IEEE} Annual Symposium on Foundations of Computer
  Science, {FOCS} 2021, Denver, CO, USA, February 7-10, 2022}, pages
  1124--1134. {IEEE}, 2021.

\bibitem[Mad16]{Mad16}
Aleksander Madry.
\newblock Computing maximum flow with augmenting electrical flows.
\newblock In Irit Dinur, editor, {\em {IEEE} 57th Annual Symposium on
  Foundations of Computer Science, {FOCS} 2016, 9-11 October 2016, Hyatt
  Regency, New Brunswick, New Jersey, {USA}}, pages 593--602. {IEEE} Computer
  Society, 2016.

\bibitem[PSY22]{PSY22}
Seth Pettie, Thatchaphol Saranurak, and Longhui Yin.
\newblock Optimal vertex connectivity oracles.
\newblock In Stefano Leonardi and Anupam Gupta, editors, {\em {STOC} '22: 54th
  Annual {ACM} {SIGACT} Symposium on Theory of Computing, Rome, Italy, June 20
  - 24, 2022}, pages 151--161. {ACM}, 2022.

\bibitem[She13]{She13}
Jonah Sherman.
\newblock Nearly maximum flows in nearly linear time.
\newblock In {\em 54th Annual {IEEE} Symposium on Foundations of Computer
  Science, {FOCS} 2013, 26-29 October, 2013, Berkeley, CA, {USA}}, pages
  263--269. {IEEE} Computer Society, 2013.

\bibitem[Vas09]{Vas09}
Virginia Vassilevska.
\newblock Efficient algorithms for clique problems.
\newblock {\em Inf. Process. Lett.}, 109(4):254--257, 2009.

\bibitem[vdBLN{\etalchar{+}}20]{van2020bipartite}
Jan van~den Brand, Yin~Tat Lee, Danupon Nanongkai, Richard Peng, Thatchaphol
  Saranurak, Aaron Sidford, Zhao Song, and Di~Wang.
\newblock Bipartite matching in nearly-linear time on moderately dense graphs.
\newblock In Sandy Irani, editor, {\em 61st {IEEE} Annual Symposium on
  Foundations of Computer Science, {FOCS} 2020, Durham, NC, USA, November
  16-19, 2020}, pages 919--930. {IEEE}, 2020.

\bibitem[WW18]{WW10}
Virginia~Vassilevska Williams and R.~Ryan Williams.
\newblock Subcubic equivalences between path, matrix, and triangle problems.
\newblock {\em J. {ACM}}, 65(5):27:1--27:38, 2018.

\end{thebibliography}
\end{document}